\documentclass[10pt]{article}
\usepackage{etex}

\usepackage{easybmat}
\usepackage{enumerate}
\usepackage{float}
\usepackage[utf8]{inputenc}

\usepackage{amssymb,amstext,amsmath,mathrsfs,array,booktabs,amsthm}
\usepackage{bm}
\usepackage{verbatim,amsbsy}
\usepackage{psfrag, epstopdf}

\usepackage{graphicx,wrapfig}
\usepackage[hang]{subfigure}
\usepackage[all]{xy}

\usepackage{url}
\usepackage[hypertexnames=false]{hyperref}
\usepackage{slashbox, multirow}

\usepackage[numbers,sort&compress]{natbib}

\newtheorem{scheme}{Scheme}
\newtheorem{lemma}{Lemma}
\newtheorem{theorem}{Theorem}
\newtheorem{definition}{Definition}
\newtheorem{corollary}{Corollary}[theorem]
\newtheorem{obs}{Observation} 
\newtheorem{clm}{Claim} 

\title{Nucleation-free $3D$ rigidity}

\author{
Jialong~Cheng $^\star$,
Meera~Sitharam \thanks{CISE, University of Florida, Gainesville, FL, USA. ({\tt {jicheng,sitharam}@cise.ufl.edu})},
Ileana~Streinu \thanks{Computer Science Department, Smith College, Northampton, MA, USA.({\tt istreinu@smith.edu,streinu@cs.smith.edu})}
}

\begin{document}

\maketitle


\def\R{\mathscr{R}_k}
\def\P{\mathscr{P}}
\def\QED{\ensuremath{{\square}}}
\def\markatright#1{\leavevmode\unskip\nobreak\quad\hspace*{\fill}{#1}}
\def\a {\textbf{a}}
\def\b {\textbf{b}}
\def\A {\textbf{u}}
\def\B {\textbf{v}}
\def\c {\textbf{c}}
\def\p {\mathbf{p}}
\def\q {\mathbf{q}}
\def\I {\mathscr{I}}
\def\X {\mathcal{X}}





\begin{abstract}
When all non-edge distances of a graph realized in $\mathbb{R}^{d}$ as a {\em bar-and-joint framework} are generically {\em implied} by the bar (edge) lengths, the graph is said to be {\em rigid} in $\mathbb{R}^{d}$. For $d=3$, characterizing rigid graphs, determining implied non-edges and {\em dependent} edge sets remains an elusive, long-standing open problem. 

One obstacle is to determine when implied non-edges can exist without non-trivial rigid induced subgraphs, i.e., {\em nucleations}, and how to deal with them.

In this paper, we give general inductive construction schemes and proof techniques to generate {\em nucleation-free graphs} (i.e., graphs without any nucleation) with implied non-edges. As a consequence, we obtain (a) dependent graphs in $3D$ that have no nucleation; and (b) $3D$ nucleation-free {\em rigidity circuits}, i.e., minimally dependent edge sets in $d=3$. It additionally follows that true rigidity is strictly stronger than a tractable approximation to rigidity given by Sitharam and Zhou \cite{sitharam:zhou:tractableADG:2004}, based on an inductive combinatorial characterization. 

As an independently interesting byproduct, we obtain a new inductive construction for independent graphs in $3D$. Currently, very few such inductive constructions are known, in contrast to $2D$.

\end{abstract}

\pagestyle{myheadings}
\thispagestyle{plain}
\markboth{J Cheng, M Sitharam and I Streinu}{Nucleation-free $3D$ rigidity}


\section{Introduction}
\label{sec:introduction}
{\bf Combinatorial rigidity in $3D$: obstacles.}
A {\em bar-and-joint framework}, or {\em framework} $G(\p)$ in $\mathbb{R}^{d}$ is a graph $G=(V, E)$ together with a mapping of its vertices to a set of points $\p$ in $\mathbb{R}^{d}$. Intuitively, we say a framework is {\em rigid} if in a small enough neighborhood of $G(\p)$, every framework $G(\q)$ with the same edge lengths as $G(\p)$ is {\em congruent} to $G(\p)$, i.e., it additionally has the same non-edge lengths as $G(\p)$.
The {\em rigidity matrix} of a framework $G(\p)$ in $\mathbb{R}^{d}$ where $G$ has $n$ vertices and $m$ edges is a matrix with $m$ rows and $nd$ columns. Each row corresponds to an edge and each column corresponds to a coordinate of a vertex. A framework $G(\p)$ is {\em generic} in $\mathbb{R}^{d}$ if its rigidity matrix has maximum rank over all frameworks of $G$ in $\mathbb{R}^{d}$ in which case, we refer to it as a {\em generic rigidity matrix} for the graph $G$. A set of edges $E^\prime \subseteq E$ is {\em independent} if their corresponding rows in the rigidity matrix are linearly independent in a generic framework. A graph is {\em rigid} (resp. {\em flexible}) if it has a generic framework that is rigid (resp. not rigid). A graph is {\em minimally rigid} if it is rigid, and the removal of {\em any} edge makes it flexible. 

\medskip\noindent
Finding combinatorial characterizations for generic rigidity and independence of for $d\ge 3$ is an elusive, long-standing open problem. James Clerk Maxwell's work from the 19th century gives a necessary condition for independence, resp. minimal rigidity:

\medskip
\noindent
{\bf Maxwell's Counting Condition \cite{maxwell:equilibrium:1864}} 
{\em 
A graph $G$ satisfies Maxwell's counts in $\mathbb{R}^{d}$ if
all of its subgraphs $G^\prime = (V', E')$ contain at most $|E^\prime|\leq d|V^\prime|-{d+1 \choose 2}$ edges. In addition, {\em minimally rigid} graphs in $\mathbb{R}^{d}$ must have exactly $|E|=d|V|-{d+1 \choose 2}$ edges. Here ${d+1 \choose 2}$ is the number of degrees-of-freedom ({\em dof} for short, which is the minimum number of edges needed to be rigid) of a rigid body in $\mathbb{R}^{d}$.
}

\medskip
\noindent
In $2D$, Laman's Theorem shows that Maxwell's counting condition is sufficient, with several other equivalent characterizations \cite{recski:networkRigidity:1984,lovasz:yemini:rigidity:1982,tay:proofLaman:1993} leading to efficient algorithms.

\medskip
\noindent
{\bf Laman's Theorem \cite{laman:rigidity:1970}}
{\em
A graph $G=(V,E)$ is independent in $2D$ if and only if
every subgraph $G^\prime=(V^\prime, E^\prime)$ of $G$ satisfies the edge-sparsity counts $|E^\prime| \leq 2|V^\prime|-3$. If, in addition, $|E| = 2|V|-3$, then $G$ is minimally rigid.}

\medskip
\noindent 
However, a first obstacle to combinatorial characterization of rigidity in dimensions $d\ge 3$ is that Maxwell's counting condition is not
sufficient, i.e., there are graphs that satisfy Maxwell's counting condition but are dependent.

\medskip
\noindent 
{\bf Note:} the remainder of the paper deals with $d=3$.

\medskip
\noindent 
Graphs that satisfy Maxwell's condition in $3D$ are called {\em $(3, 6)$-sparse} (sometimes called {\em Maxwell-independent}).

\medskip
\noindent
A classical example illustrating insufficiency of Maxwell's counting condition in $3D$ is the so-called {\it double-banana
graph} in Fig.~\ref{fig:doubleBanana}. It satisfies
Maxwell's counting condition, but the graph is clearly flexible, and
dependent. The reason is that since each {\em banana} (a $K_5$ with 1 edge missing) is rigid as an induced subgraph, the distance along the non-edge $\{a, b\}$ shared by the $2$ bananas is determined by each banana. When the distance along a non-edge $\{a, b\}$ is determined by a graph $G$, i.e., linearly dependent on the rows of $G$'s generic rigidity matrix, then $\{a, b\}$ is called an {\em implied non-edge} \footnote{ The concept of {\em implied non-edge} is similar to but weaker than the notion of {\em globally linked
pair} from \cite{jackson:jordan:szabadka:GloballyLinked:2006}, which refers to a pair of vertices whose distance is generically fixed by the graph. The distance associated with an implied non-edge is in contrast
generically restricted to finitely many values. An alternative name
for an implied non-edge could be a {\em linked pair}.} in $G$. In this case, since $\{a, b\}$ is implied by both bananas, the graphs $G$ is {\em dependent}.

\begin{center}
\begin{figure}
\begin{center}
\includegraphics[width=0.5\textwidth]{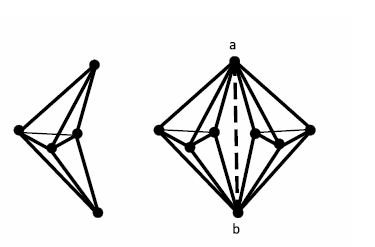} 
\end{center}
\caption{On the left is the ``banana'' graph, which is a $K_5$ (complete graph on five vertices) with one edge missing; on the right is a double-banana, which consists of two bananas gluing together along their respective non-edge and the non-edge is implied.}\label{fig:banana}\label{fig:doubleBanana}\label{fig:tripleBanana}
\end{figure}
\end{center}

\medskip\noindent
Another example of a dependent graph illustrating insufficiency of Maxwell's counting condition, i.e., being dependent while satisfying Maxwell's counting condition, is due to Crapo (Fig.~\ref{fig:hinge}). The graph contains a so-called {\em hinge}, i.e., a pair of vertices common to at least two {\em rigid components}, where a rigid component is a maximal subset $S$ of vertices of a graph $G$ such that the non-edges in $S$ are implied by the edges in $G$, possibly outside the graph induced by $S$. For example, the $\{a, b\}$ pair in Fig.~\ref{fig:hinge} is a hinge. Again, here hinges are implied non-edges that are the causes of the insufficiency of Maxwell's condition. 

\begin{center}
\begin{figure}
\begin{center}
\scalebox{0.4}[0.4]{\includegraphics{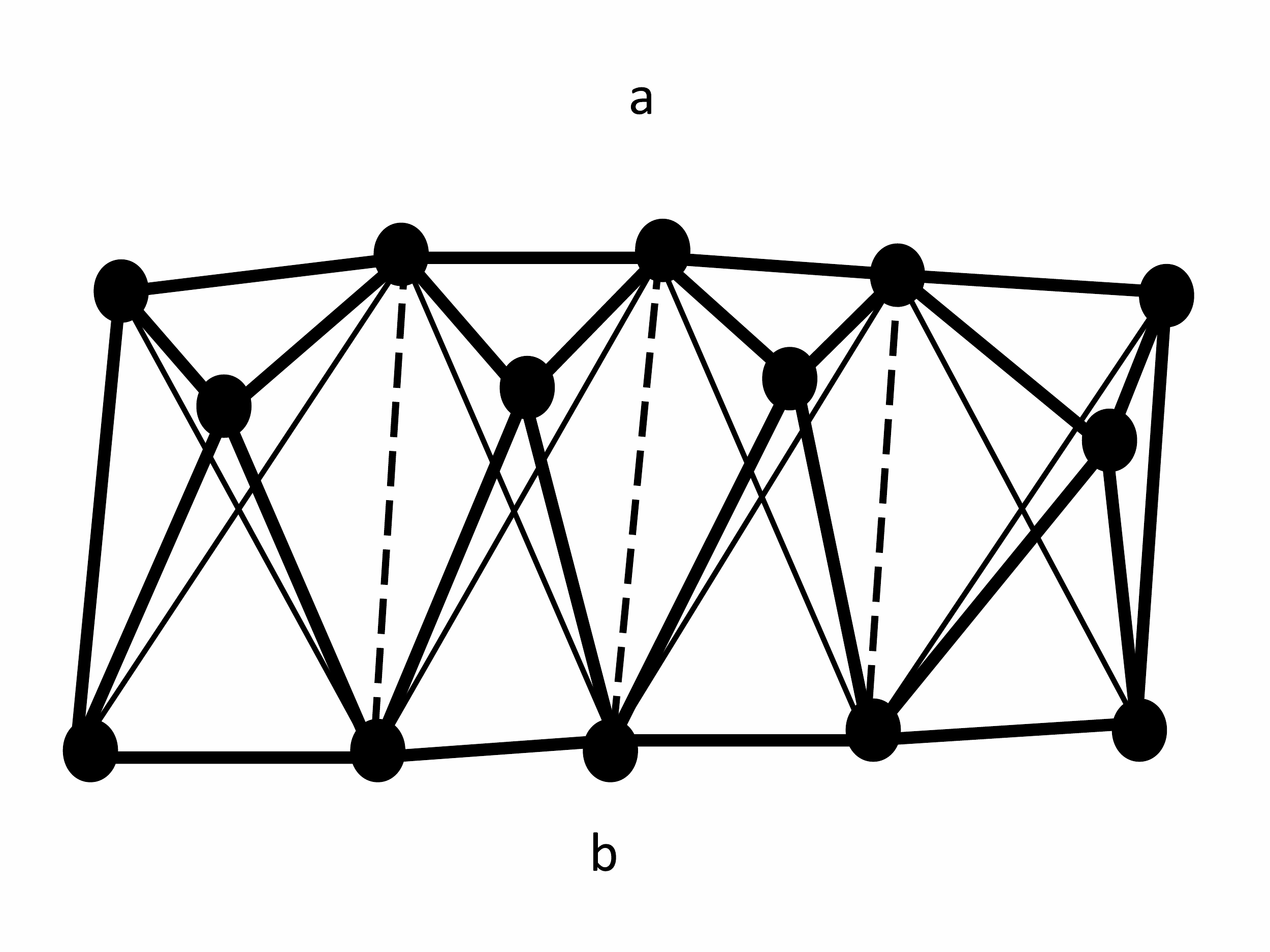}}
\end{center}
\caption{Crapo's graph with a ``hinge" structure: $\{a, b\}$ is a hinge, since it is shared by two rigid components.}\label{fig:hinge}
\end{figure}
\end{center}

\medskip \noindent
All known examples of insufficiency of Maxwell's counting condition have implied non-edges
and implied non-edges play an important role in the insufficiency of Maxwell's counting condition.

\medskip \noindent
However, an important observation in the above examples is that although some implied non-edges lie inside rigid components as opposed to rigid induced subgraphs, these ``troublesome'' double-implied non-edges exist due to the presence of a rigid induced subgraph somewhere in the graph. I.e., the above examples satisfy the following property:

\medskip
\noindent
{\bf Nucleation property.} A graph $G$ has the {\em nucleation
property} if it contains a non-trivial rigid subgraph, i.e., a rigid subgraph in isolation, which we call a {\em rigid nucleus}. Here, we use ``trivial'' to refer to graphs with $4$ or fewer vertices. If a graph does not have any nucleus, we call it {\em nucleation-free}.

\medskip
\noindent
Note that in $2D$, every implied non-edge in fact lies inside a nucleation, as a straightforward consequence of Laman's Theorem.

\medskip
\noindent
In $3D$, provided a graph has nucleation property, there is a potential method of overcoming the obstacle of implied non-edges using $(3, 6)$-sparsity for detecting dependence \cite{sitharam:zhou:tractableADG:2004}: recursively identify nucleations, add non-edges in those graphs to complete them, and then check $(3, 6)$-sparsity in other parts of the graph.

\medskip
\noindent
However, for nucleation-free graphs, the approach in \cite{sitharam:zhou:tractableADG:2004} collapses to simple $(3, 6)$-sparsity check, leading to the second obstacle to the problem of combinatorial characterizations of $3D$ rigidity. In particular, in a nucleation-free {\em dependent} graph, the approach in \cite{sitharam:zhou:tractableADG:2004} would fail in that it cannot detect the implied non-edges, since it relies on a nucleation as a starting point. The existence of nucleation-free dependent graphs indicates that a gap exists between {\em module-rigidity} proposed in \cite{sitharam:zhou:tractableADG:2004} and true rigidity. Thus, to better understand rigidity, we need to understand the obstacle posed by nucleation-free graphs with implied non-edges. As a first step towards this goal, it is natural to ask: 

\medskip\noindent
{\bf Problem ($\star$) : General inductive construction schemes for nucleation-free  graphs with implied non-edges.} {\em How do we construct general families of nucleation-free graphs that have implied non-edges?}

\medskip
\noindent
{\bf Main Contributions.}
We provide general inductive construction schemes and proof techniques answering Problem ($\star$). We give systematic classifications of proof ingredients needed for our proof techniques, and those ingredients can be generalized, mixed and matched to generate and validate construction schemes. In addition, we give several example graphs that satisfy the starting graph requirements for our general inductive construction schemes for Problem ($\star$).

As a byproduct of one of our schemes, we find an inductive method to construct nucleation-free, independent graphs. It should be noted that there are very few general inductive construction for independent graphs in $3D$. The only known ones are vertex split and Henneberg constructions(\cite{bib:TayWhiteley85}, \cite{WhiteleyVertexSplitting1990}), which are distinctly different and cannot generally mimic our inductive construction.

\medskip
\noindent
{\bf Further consequences.} 
As mentioned earlier, we show a gap exists between {\em module-rigidity} proposed in \cite{sitharam:zhou:tractableADG:2004} and true rigidity. 
The next step for a better understanding of rigidity should be to find notions that detect implied non-edges and dependence in nucleation-free graphs.
As another consequence of our work, we show the first general families of examples of flexible $3D$ rigidity circuits with no nucleations. In contrast, all rigidity circuits in $2D$ are rigid. Until now, flexible rigidity circuits without nucleation were only available in $4D$, but not in $3D$ (see discussions earlier and in \cite{graver:servatius:rigidityBook:1993}): {\it The only known non-rigid circuits in the $3D$ rigidity matroid arise from amalgamations of circuits forced by Maxwell's counting condition.} 
 Our result implies, in addition, that Lovasz' characterization \cite{lovasz:yemini:rigidity:1982} of $2D$ rigidity via {\em coverings} cannot be extended to $3D$.

\medskip
\noindent
{\bf Note:} Nucleation-free rigidity circuits with implied non-edges have been conjectured and written down by many (\cite{taybar:1993}, \cite{JacksonJordanrank:2006}). However, to the best of our knowledge we are the first to give proofs. In particular, in \cite{taybar:1993}, Tay claimed a class of flexible rigidity circuits without any nuclei. One of his examples, {\em $n$-butterflies}, in which he claimed existence of implied non-edges, is the same as our warm-up example graphs, {\em ring of roofs}. But his proof attempt has a serious gap which we will describe in detail in Appendix \ref{sec:Taycounter}.

\medskip
\noindent
{\bf Organization of the paper.} In Section \ref{sec:first}, we give a warm-up example of graphs with no nucleation and give two different proof techniques to show the existence of implied non-edges. In Section \ref{sec:natural}, we give {\em straightforward extensions} of our warm-up construction and proof techniques.

In Section \ref{sec:indep}, we give a significantly more powerful inductive construction scheme, {\em roof-addition}, for nucleation-free (independent) graphs with implied non-edges. The independence is shown in Theorem \ref{thm:inductive} and is a stand-alone result for inductive construction of independent graphs. The other two properties (nucleation-free and presence of implied non-edges) are shown in Theorem \ref{thm:secondInductiveApproach}. In particular, in Observation \ref{obs:other} we list and show the properties of several example {\em starting graphs} for the roof-addition scheme. In Theorem \ref{thm:starting} how we can inductively use starting graphs to generate nucleation-free, independent graphs with implied non-edges.

In Section \ref{sec:dependent} (Theorem \ref{thm:joinIndep}), we show how to obtain nucleation-free (minimally) dependent graphs.
Then we exhibit a family of graphs in Corollary \ref{cor2} to show that the algorithm of \cite{sitharam:zhou:tractableADG:2004} i.e., characterization of module-rigidity is distinct from true rigidity. In Section \ref{sec:conclusions} we give open problems.



\section{Warm-up Example: Rings of Roofs}\label{sec:first}
In this section, we give a warm-up family of examples to motivate construction schemes and proof techniques for Problem ($\star$). The first proof technique, which we call {\em flex-sign}, is to {\em directly} show the existence of implied non-edges; the second proof technique, which we call {\em rank-sandwich} technique, needs the independence of the graph. There are two ingredients in the rank-sandwich proof technique that
generate stand-alone results: (1) show the independence of the graph (2) show the rank upper bound of the graph with implied non-edges added (which in turn can be shown using two different methods). 

As noted earlier, \cite{taybar:1993} claims the existence of implied non-edges in this family essentially without proof. The {\em basic building block} for this warm-up family of nucleation-free graphs is a {\em roof} (called a ``butterfly'' in \cite{taybar:1993}).

\medskip
\noindent
{\bf Roof.} A \emph{roof} is a graph obtained from $K_5$, the complete graph on five vertices, by deleting two non-adjacent edges. These are called the {\em hinge non-edges} of a roof. A 3D realization of a roof is illustrated in Fig.~\ref{fig:roof}. In the terminology of \cite{streinu:whiteley:origami:2005}, this is a {\em single-vertex origami} over a $4$-gon.

\medskip
\noindent
{\bf Example (Rings of roofs)}: A {\em ring} of roofs is constructed as follows. Exactly two roofs share a hinge non-edge, as in Fig.~\ref{fig:twoRoofs}. Each roof then shares hinge non-edges with at most two others. Such a chain of seven or
more roofs is closed back into a ring, as depicted 
in Fig.~\ref{fig:ring}.

We will denote by $\mathscr{R}_k$ a ring of $k$ roofs. $\mathscr{R}_k$ is rigid for $k\le 6$. It is easy to see that there is no nucleus in $\mathscr{R}_k$ for $k\ge 7$. We will show that implied non-edges exist in $\mathscr{R}_k$ for any $k\ge 7$, using two different proof techniques.

\begin{center}
\begin{figure}[!htbp]
\centering
\scalebox{0.2}[0.2]{\includegraphics{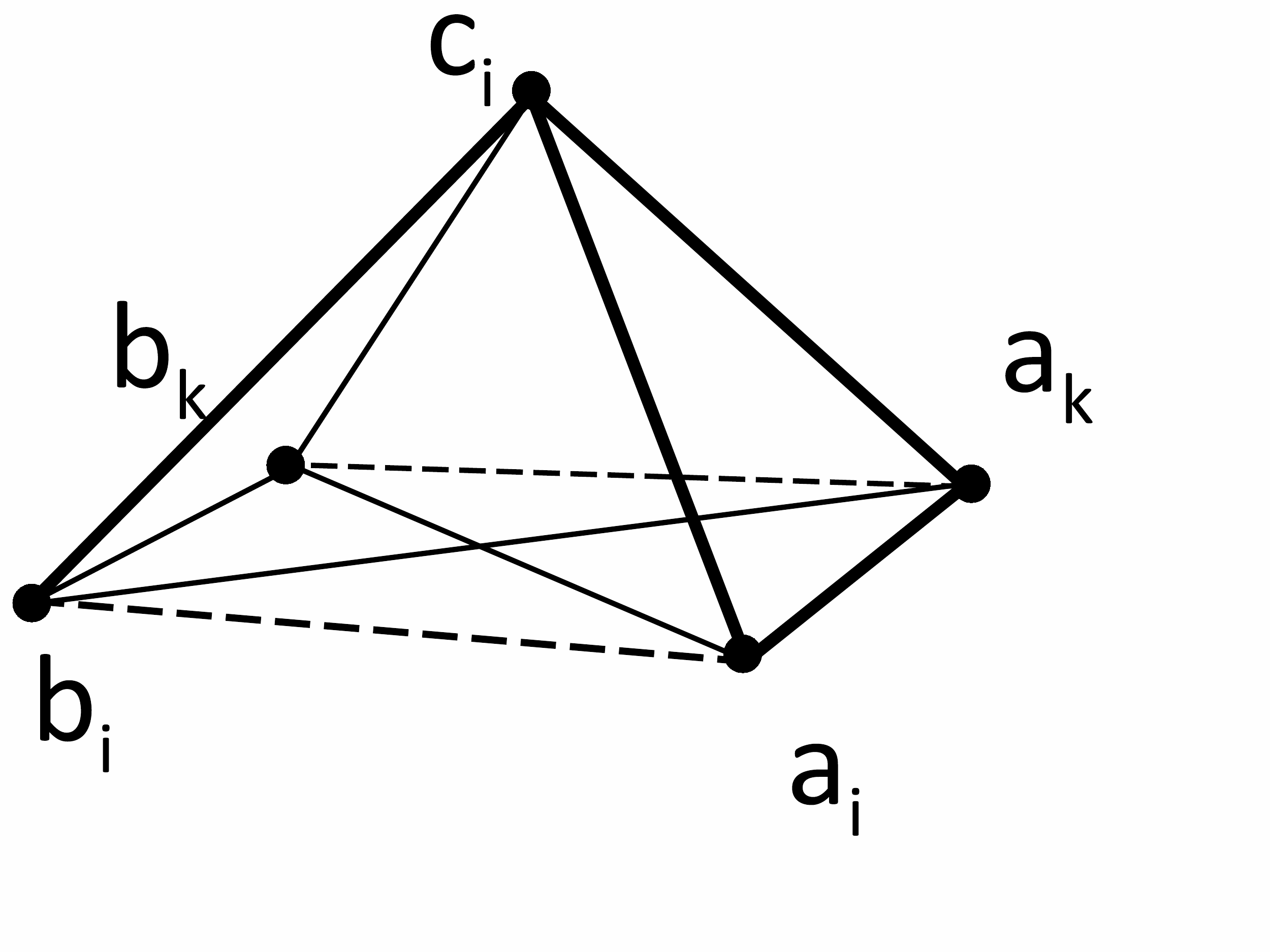}\hspace{-2cm}\includegraphics{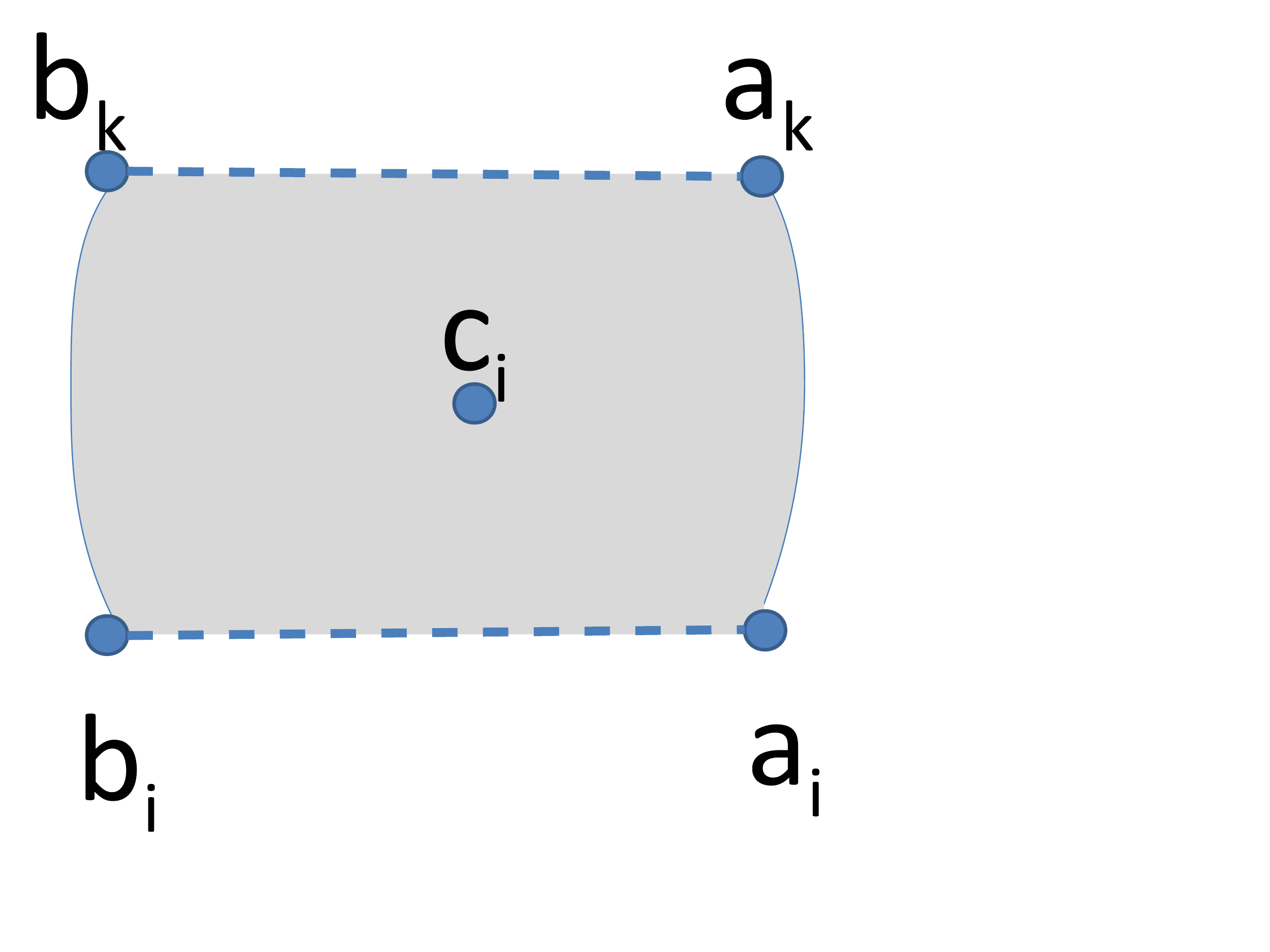}}\caption{A roof: a $K_5$ (complete graph on $5$ vertices) with two non-adjacent edges missing. On the left we give the geometric structure of the roof in space. On the right we give a schematic of a roof: the bar-and-joint structure is not shown for clarity, but the positions hinges are schematically depicted with $a_i$, $b_i$, $a_{k}$, $b_{k}$. The vertex $c_i$ may or may not be shown in a schematic use of the roof later.}\label{fig:roof}
\end{figure}
\end{center}

\begin{center}
\begin{figure}[!htbp]
\begin{center}
\scalebox{0.2}[0.2]{\hspace{-2cm}\includegraphics{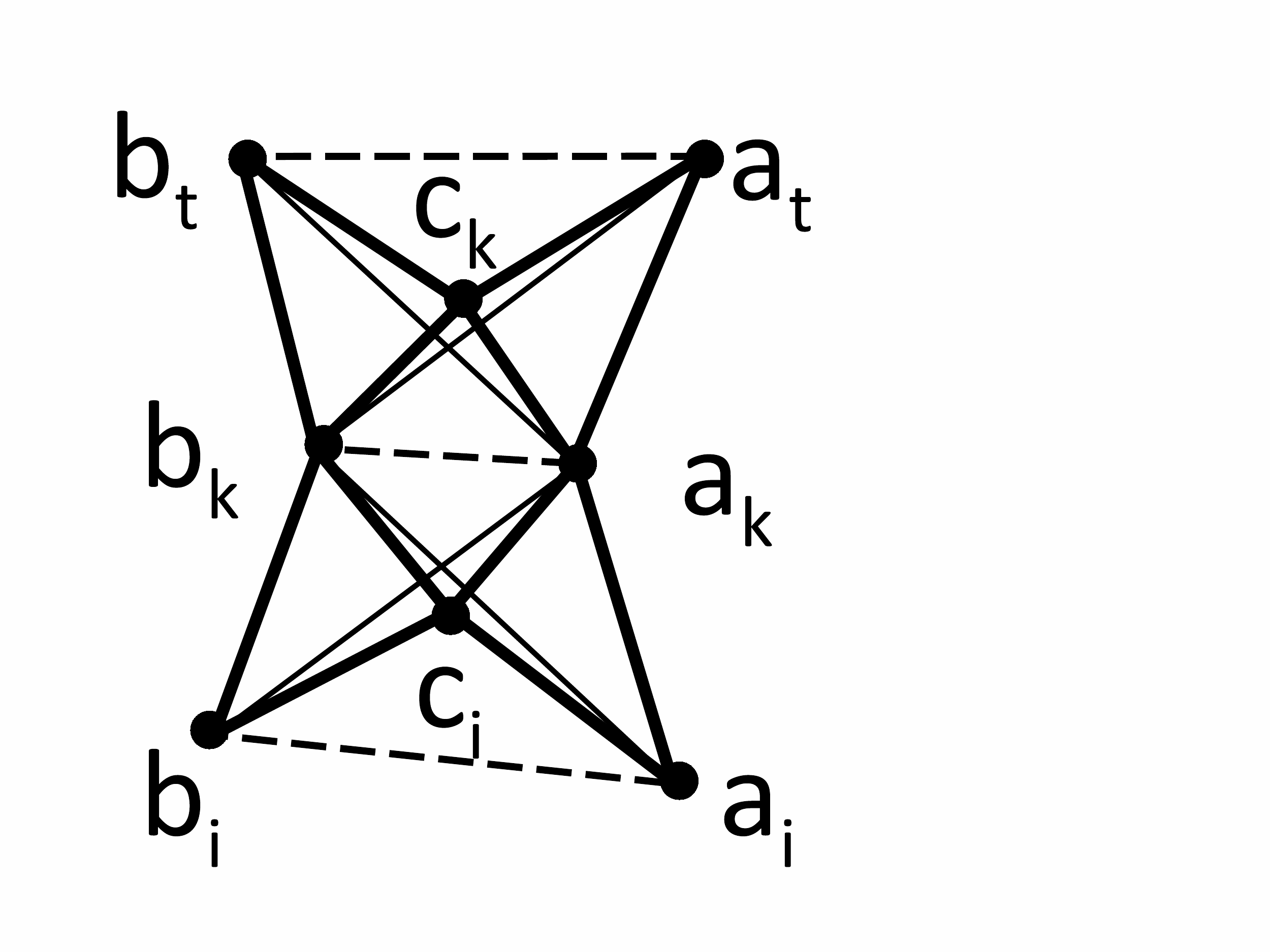}\hspace{-6cm}\includegraphics{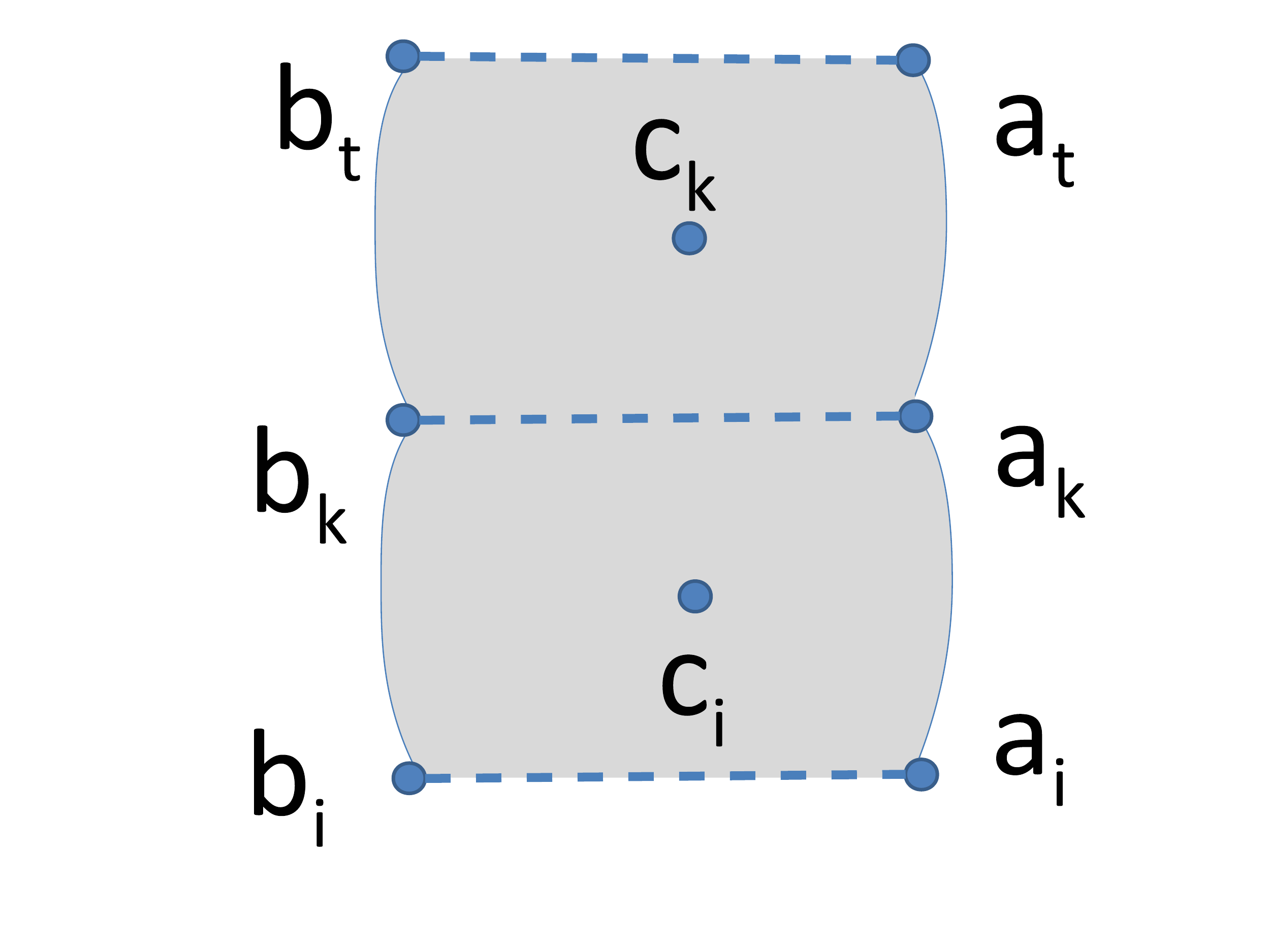}}
\end{center}
\caption{Connecting two roofs. On the left is the real geometric connection, while on the right is a schematic showing how the two roofs are connected via a hinge}\label{fig:twoRoofs}
\end{figure}
\end{center}

\begin{center}
\begin{figure}[!htbp]
\begin{center}
\scalebox{0.4}[0.4]{
\begin{tabular}{c}
\includegraphics[width=1.2\textwidth]{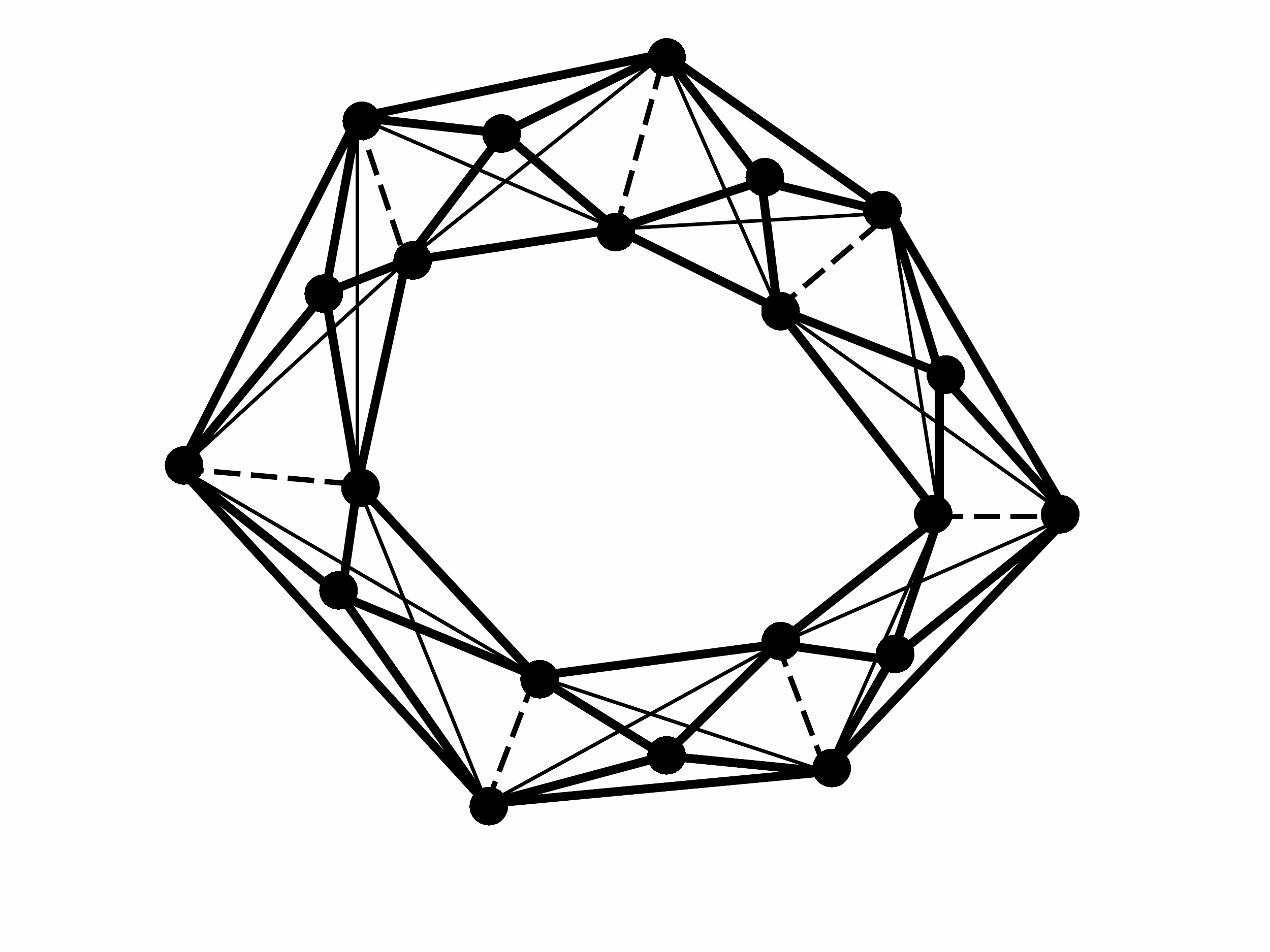}
\includegraphics[width=1.2\textwidth]{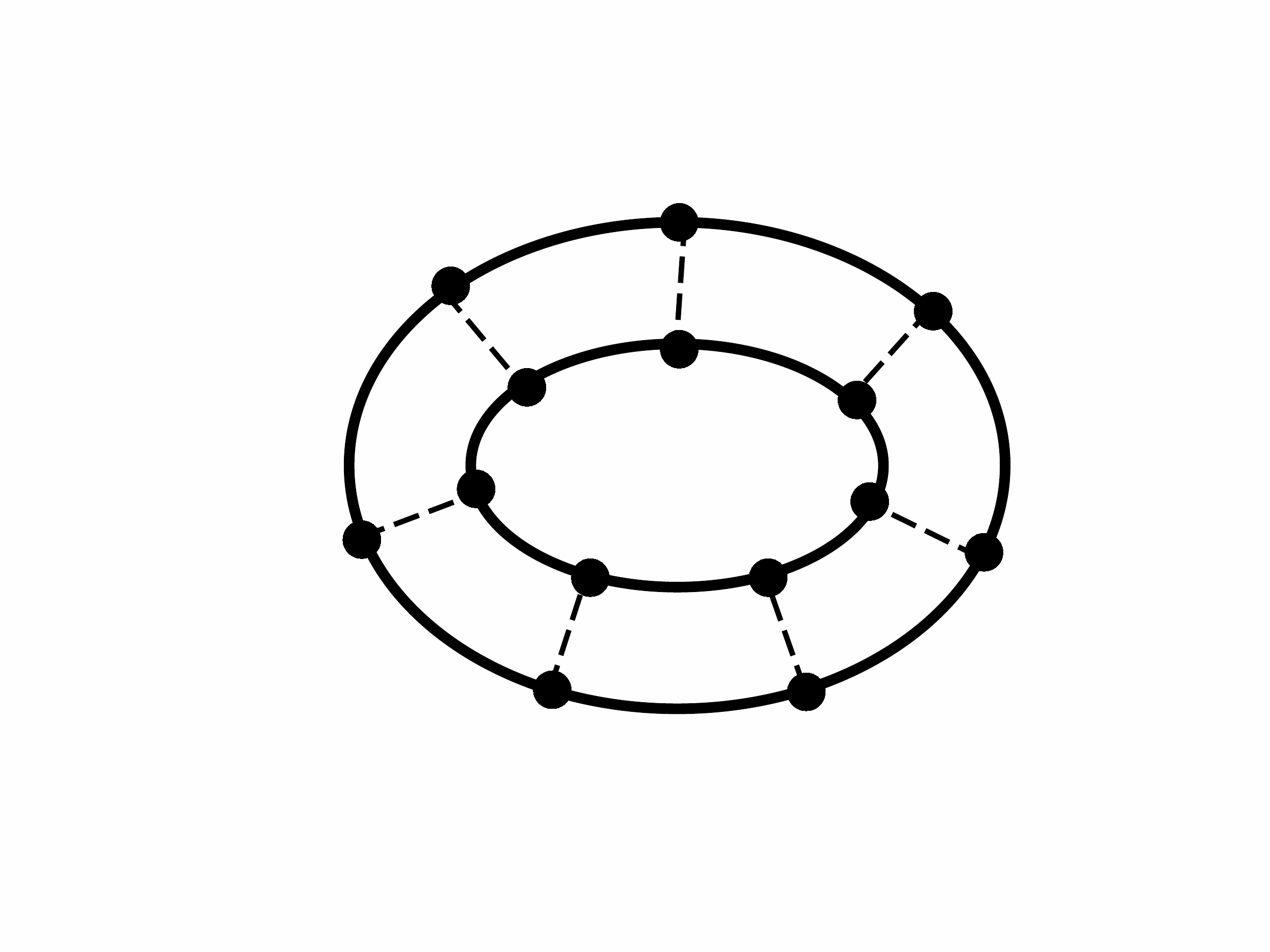}
\end{tabular}}
\end{center}
\caption{A ring of $7$ roofs: connecting $7$ roofs in the manner shown in Fig. ~\ref{fig:twoRoofs} and we can see each roof can
be thus connected to at most two others. Such a chain of seven roofs is closed back into the ring shown here. On the left is the geometric structure of the ring and on the right is the schematic of the ring.}\label{fig:ring}
\end{figure}
\end{center}

\subsection{Flex-sign technique for existence of implied non-edges in rings of roofs}\label{sec:direct}
In this section, we give our first proof technique, called {\em flex-sign technique}, for the existence of implied non-edges in rings of roofs. This proof technique relies on the infinitesimal properties of single-vertex origamis from \cite{streinu:whiteley:origami:2005}, together with expansion/contraction properties of convex polygons \cite{connelly:rigidityAndEnergy:1982} and pointed pseudo-triangulations \cite{streinu:pseudoTriang:dcg:2005}, applied to the simplest possible case of a 4-gon. These results show that the roof realizations in the {\em ring framework $\mathscr{R}_k(\p)$} from Fig.~\ref{convex} have the expansion/contraction properties stated in the caption. We show that the hinge non-edges are implied in any ring framework $\mathscr{R}_k(\p)$ of $k\geq 7$ roofs consisting of $1$ convex and $k-1$ pointed pseudo-triangular roofs. This expansion/contraction property is similar to the {\em squeeze} in \cite{taybar:1993} (see Appendix \ref{sec:Taycounter}), i.e., a non-trivial flex or motion along the hinge non-edge.

\begin{center}
\begin{figure}
\begin{center}
\scalebox{0.2}[0.2]{\includegraphics{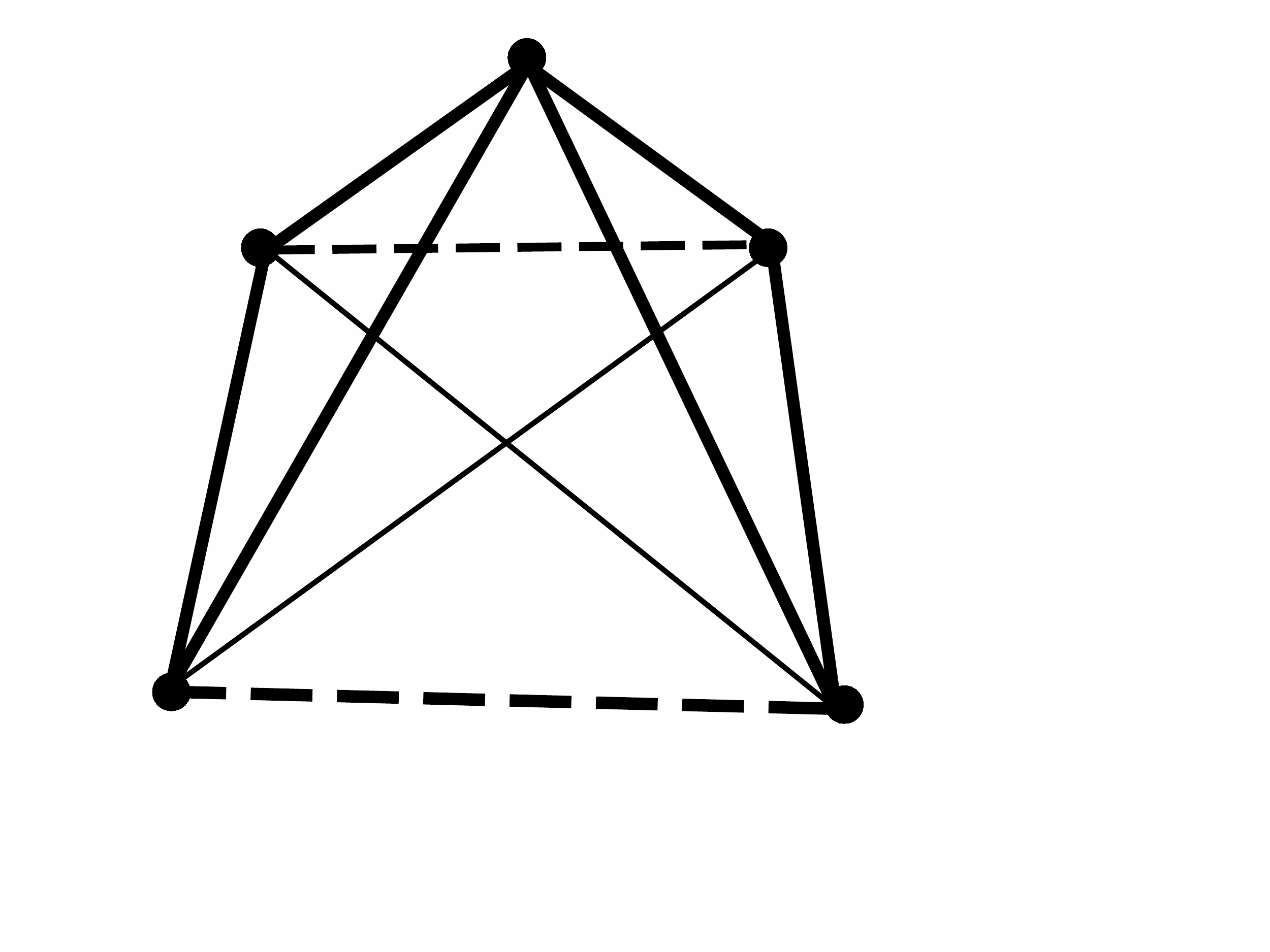}\hspace{-6cm}\includegraphics{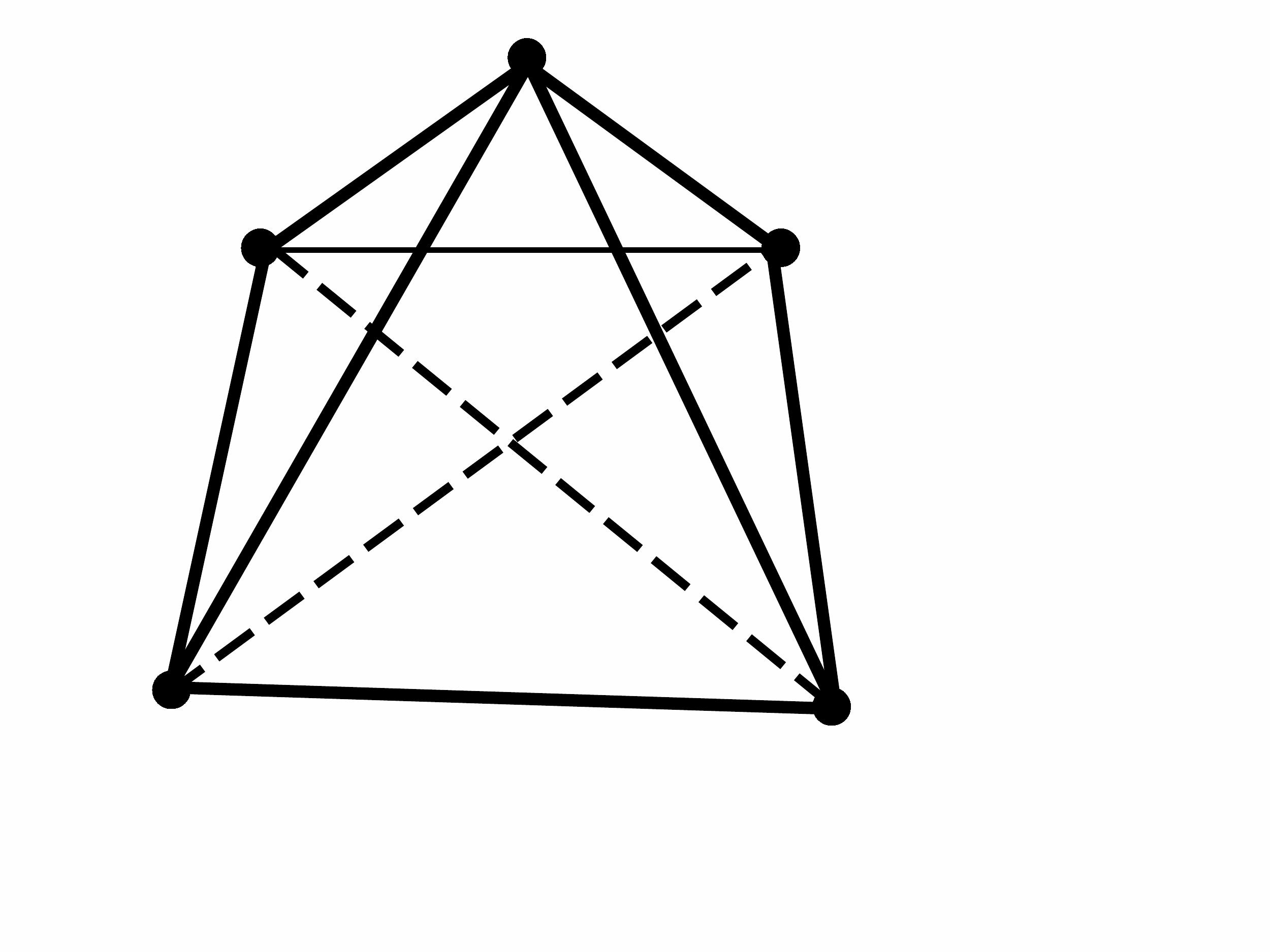}\hspace{-6cm}\includegraphics{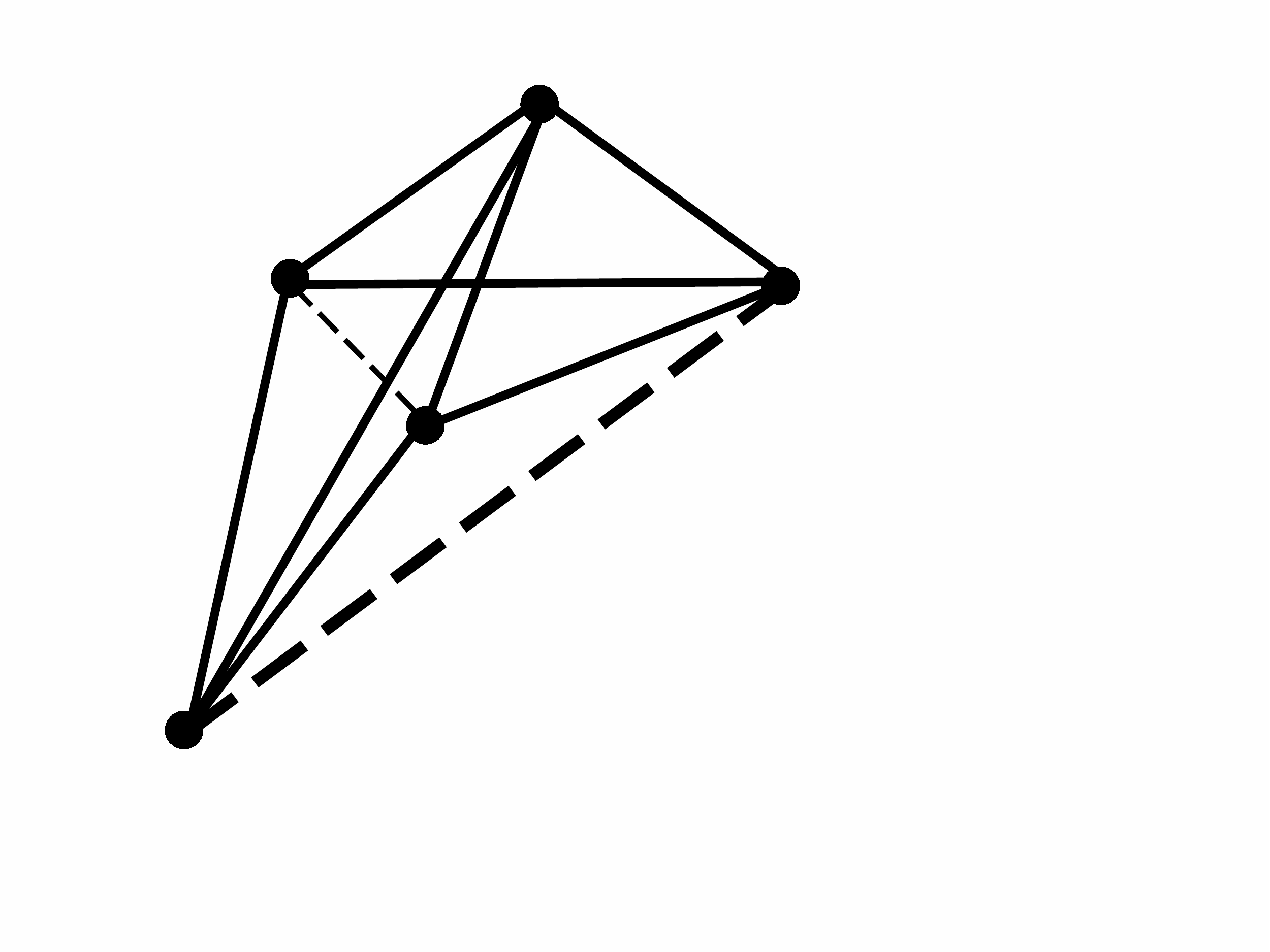}}
\end{center}
\caption{All 3 types of roofs could occur in the warm-up construction, and the proof technique in Section \ref{ringupper3} can be applied to all types. However, the figure on the left is a roof whose base is ``crossing'' in the shape of a butterfly. The proof technique in Section \ref{sec:direct} does not apply to this type of roof. The figure in the middle is a convex or expansive-contractive roof, i.e., if one of its hinge non-edges has a contractive motion, then the other is forced to have an expansive motion. The figure on the right is a pointed pseudo-triangular or expansive-expansive roof. Its
two hinge non-edges move in either both expansive or both
contractive fashion.} \label{convex}\label{pseudoPT}
\end{figure}
\end{center}

\begin{lemma}\label{lem:step2third}
For all ring frameworks $\mathscr{R}_k(\p)$ of $k-1$, pointed pseudo-triangular roofs and one convex roof, the hinge non-edges are implied.
\end{lemma}

\begin{proof}
Assume that there exists an infinitesimal motion along one hinge non-edge $\{a, b\}$, and without loss of generality, assume that motion is
expansive. Then the increase/decrease
patterns of the two hinge non-edges of pointed pseudo-triangular roofs, i.e., expansive-expansive and the convex roof, i.e., expansive-contractive, when followed along
the ring back to the starting hinge non-edge, imply that the motion of $\{a, b\}$ is contractive, a contradiction.
\end{proof} 

\noindent
Next we show that in fact, the hinge non-edges are implied generically.

\begin{lemma}\label{lem:genericrk}
There are generic frameworks $\R(\p)$ as in Lemma
\ref{lem:step2third}.
\end{lemma}

\begin{proof}
Any framework $\R(\p)$  where the first roof is strictly convex
and the remainder are strictly pointed pseudo-triangular as in Lemma
\ref{lem:step2third},  can be viewed as a point in $\mathbb{R}^{3*3k}$.
There is an open neighborhood around $\R(\p)$ consisting of frameworks that continue to have the same
convex/pointed pseudo-triangular property.
This shows that the hinge non-edges are implied generically.
\end{proof}

\noindent
Combining Lemmas \ref{lem:step2third} and \ref{lem:genericrk}, it
follows that the hinge non-edges are implied in the ring graph $\R$.

\subsection{Rank-sandwich technique for the existence of implied non-edges in ring of roofs }\label{sec:impliednonedges}\label{ringupper3}
In this section, we give our second proof technique, called {\em rank-sandwich technique}, to show the existence of implied non-edges in ring of roofs $\R$. As mentioned earlier, this proof technique has two ingredients that are of interest as stand-alone techniques: 
\begin{itemize}
\item showing the independence of the graph $G$ (i.e., number of edges is the rank);
\item showing that a simple rank upper bound after adding potential implied non-edges $F$ as edges to $G$ is equal to the number of edges in $G$.
\end{itemize}  
Together this proves that the non-edges in $F$ are implied.

More specifically, we show in Section \ref{sec:indepring} that rings of roofs $\R$ are independent. Then we give two different arguments to show the rank upper bound on a ring of roofs $\R$ with hinge non-edges added. Those two arguments are {\em 2-thin cover argument} in Section \ref{ringupper1} and {\em body-hinge argument} in Section \ref{ringupper2}.

\subsubsection{Independence of ring of roofs}\label{sec:indepring}
We first consider a ring of tetrahedra, i.e., ring of $K_4$'s where neighboring $K_4$'s share an edge.
\begin{obs}\label{obs:ringK4}
A ring of at least six tetrahedra is independent.
\end{obs}

A {\em Henneberg-II construction} on a graph $G$ in $\mathbb{R}^{3}$ is to first
choose a vertex set $W=\{w_1, w_2$, $w_3, w_4\}$ of $4$ vertices of $G$ such that there is at least $1$ edge induced by $W$, deleting $1$ edge between the vertices of $W$, and then adding a new vertex $v$ and four edges $(v,w_1)$, $(v,w_2)$, $(v,w_3)$, $(v,w_4)$.
From \cite{bib:TayWhiteley85}, we know the following:
\begin{theorem}[\cite{bib:TayWhiteley85}]\label{thm:Henneberg}
Henneberg-II constructions preserve independence of a graph in 3D.
\end{theorem}

\noindent
Combining Observation \ref{obs:ringK4} and Theorem \ref{thm:Henneberg}, we obtain the following:

\begin{theorem}\label{r7}
A ring of roofs $\R$  is independent for $k\ge 6$.
\end{theorem}

\subsubsection{Rank upper bound using 2-thin cover argument}\label{ringupper1}
Here, we give our first argument of the rank upper bound of rings of roofs $\mathscr{R}_k$ with implied non-edges added. First, we need the following concepts. For any graph $G$, independent edge sets of $G$ define a {\em matroid} of $G$ in $\mathbb{R}^d$, which is called the {\em generic rigidity matroid} of $G$. The {\em rank} of $G$ in $\mathbb{R}^d$ is the rank of its generic rigidity matroid in $\mathbb{R}^d$. 

We have the following lemma.

\begin{lemma}
\label{lem:step2first} The rank of the 3D rigidity matroid of a ring of roofs $\mathscr{R}_k$ does not change if we add all hinge non-edges.
\end{lemma}

\medskip \noindent
To prove this, we need to introduce a few more concepts.

\medskip \noindent
A \emph{cover} of a graph $G=(V, E)$ is a collection $\mathcal{X}$
of pairwise incomparable subsets of $V$, each of size at least two,
such that $\bigcup\limits_{X\in \mathcal{X}} E(X) = E$, where $E(X)$ is the edge set induced by $X$. A cover $\mathcal{X} =
\{X_1, X_2, \ldots, X_n\}$ of $G$ is {\em $2$-thin} if $|X_i \cap X_j |
\leq 2$ for all $1 \leq i < j \leq n$. Let the {\em shared part} $S(\mathcal{X})$ be the
set of all pairs of vertices $a, b$ such that $X_i \cap X_j = \{a,
b\}$ for some $1 \leq i < j \leq n$. For each
$\{a, b\} \in S(\mathcal{X})$, let $d(a, b)$ be the number of sets $X_i$ in $\mathcal{X}$ such that $\{a, b\} \subseteq X_i$. Let $G[X_i]$ denote the subgraph of $G$ induced by $X_i$. Sometimes induced subgraphs$\{G[X_1], G[X_2]$, $\ldots, G[X_n]\}$ are called {\em covering subgraphs} and these additionally serve as the basic building blocks of many of our constructions.

\medskip \noindent
We will rely on the following theorem by Jackson and Jord\'an
\cite{JacksonJordanrank:2006}.
\begin{theorem}\label{IPindep}(Jackson and Jord\'an \cite{JacksonJordanrank:2006})
If $\mathcal{X} = \{X_1, X_2, \ldots, X_m\}$ is a $2$-thin cover of graph
$G=(V, E)$ and subgraph $(V, S(\mathcal{X}))$ is independent, then
in 3D, the rank of the
rigidity matroid of $G^\star:= G\cup S(\mathcal{X})$, denoted as $rank(G^\star)$, satisfies the following
\begin{equation}\label{eq:IE}
rank(G^\star) \leq \sum_{X_i\in \mathcal{X}} rank(G^\star[X_i]) - \sum_{\{a, b\}\in
S(\mathcal{X})} (d(a, b) - 1),
\end{equation}
where $G^\star[X_i]$ denotes the subgraph of $G^\star$ induced by $X_i$.

\end{theorem}

\medskip
\noindent
{\bf Remark:} Both 2-thin covers and variants of the right hand side of Equation \eqref{eq:IE} have appeared in other contexts,
for example in the context of the Dress' 3D
rigidity conjectures \cite{bib:counter1, bib:counter2,
bib:counter3}, including one counterexample.  They also appear using different terminology, for example, in the context of algorithms for  geometric constraint
decomposition based on the notion of module-rigidity mentioned
earlier \cite{sitharam:zhou:tractableADG:2004},
\cite{bib:SitharamFrontier}, as well as algorithms for isolating and
efficiently solving the so-called {\em well-formed} system incidence
constraints between \emph{standard collections} of rigid bodies
\cite{bib:wellformed, bib:opt}.

\begin{proof}
(of Lemma \ref{lem:step2first}) After adding hinge non-edges into a ring $\mathscr{R}_k$ of $k$ roofs, we will get a ring $\mathscr{C}_k$ of $k$ $K_5$'s. Denote those $K_5$'s by $\{C_1, C_2, \ldots, C_k\}$ and let $\mathcal{X} = \{V(C_1), V(C_2),
\ldots, V(C_k)\}$ be a cover of $\mathscr{C}_k$. 
Note that $\mathcal{X}$ is a $2$-thin cover. Since all edges in the shared part $(V, S(\mathcal{X}))$ are disjoint, $(V, S(\mathcal{X}))$ must be independent. Applying
Theorem \ref{IPindep}, we have:
\begin{equation*}
rank(\mathscr{C}_k)\leq\sum_i^n {rank(\mathscr{C}_k[V(C_i)])}  - \sum_{\{a, b\}\in
S(\mathcal{X})} (d(a, b) - 1) = 9*k - k =8k. 
\end{equation*}
By Theorem \ref{r7}, we know a ring $\mathscr{R}_k$ of $k$ roofs is independent. Thus the rank $\mathscr{R}_k$ is equal its number of edges, which is exactly $8k$. Hence after adding all hinges to $\mathscr{R}_k$, the rank does not change.
\end{proof} 

\medskip\noindent
The proof of Lemma \ref{lem:step2first} shows how the 2-thin cover argument is used to complete the second ingredient needed to show the existence of implied non-edges: the rank upper bound for $\mathscr{R}_k$ after the hinge non-edges are added is equal to the number of edges in $\mathscr{R}_k$, which, by the independence of $\mathscr{R}_k$ shown in Theorem \ref{r7}, is equal to the rank of $\mathscr{R}_k$. Generally, if a graph $G$ satisfies the following conditions:
\begin{itemize}
\item there is a 2-thin cover $\mathcal{X}$ $=\{X_1, X_2, \ldots, X_n\}$ of $G$ such that shared part $(V, S(\mathcal{X}))$ is independent;
\item $rank(G) = \sum_{X_i\in \mathcal{X}} rank(G^\star[X_i]) - \sum_{\{a, b\}\in
S(\mathcal{X})} (d(a, b) - 1)$, where $G^\star= G\cup S(\mathcal{X})$;
\end{itemize}
then the non-edges in $S(\mathcal{X})$ are implied. These will later be used as starting requirements on graphs for inductive constructions in Scheme \ref{scheme:roof} later.

\subsubsection{Rank upper bound using body-hinge argument}\label{ringupper2}
In this section, we give a second argument of the rank upper bound of rings of roofs $\mathscr{R}_k$ with implied non-edges added. For completeness, we first observe the following:

\begin{obs}
\label{prop:step2second} Let the ring framework $\R(\p)$ be generic, then for all $i$,
the rigidity matrices of each of the banana frameworks $B_i( \p_i)$ are independent, and thus the $B_i( \p_i)$'s are rigid, where $\p_i$ is the
restriction of $p$ to the vertices in the $i^{th}$ roof $R_i$.
\end{obs}

\noindent
Together with the following result by Tay
\cite{tay:rigidityMultigraphs-I:1984} and White and Whiteley \cite{white:whiteley:algebraicGeometryFrameworks:1987} on \emph{body-hinge} structures, we can complete the proof.

\begin{theorem} \label{fact:tay}  If $\forall i\leq k$, the $i^{th}$ banana
$B_i(\p_i)$ is rigid, then the framework $\mathscr{B}_k(\p)$ is
equivalent to a body-hinge framework and is guaranteed to have at
least $k-6$ independent infinitesimal motions.
\end{theorem}

\medskip
\noindent
Observation \ref{prop:step2second} and Theorem \ref{fact:tay} show that the rank upper bound for a ring with hinge non-edges added is equal to the number of edges in the ring, thereby showing that the hinge non-edges are implied.

\section{Natural Extensions of Warm-up Example}\label{sec:natural}
In this section, we extend the warm-up example to give general constructions of nucleation-free graphs with implied non-edges. We will list some construction schemes whose correctness follows from the flex-sign or the rank-sandwich proof technique.

First, we can extend the building block to other graphs that satisfy the conditions in Lemma \ref{lem:step2third} and obtain the following construction scheme: 

\begin{scheme}[Flex-sign Ring]\label{schme:specialmotion}\hfill

{\bf Input graphs:} $G_1$, $G_2$, $\ldots$, $G_n$

\smallskip
{\bf Output graph:} a {\em ring graph} consisting of $G_1$, $G_2$, $\ldots$, $G_n$ such that (1) neighboring subgraphs $G_i$ and $G_{i+1}$ share a hinge non-edge $\{a_i, b_i\}$ with $G_n$ and $G_1$ sharing $\{a_n, b_n\}$ 
(2) each hinge non-edge is shared by exactly 2 graphs $G_i$ and $G_{i+1}$ or $G_1$ and $G_n$.

\end{scheme}

\begin{theorem}\label{thm:specialmotion}
If the input graphs $G_1$, $G_2$, $\ldots$, $G_n$ in the Flex-sign Ring scheme (Scheme \ref{schme:specialmotion}) have the following property: (1) each of the $G_i$'s is nucleation-free and (2) one of the $G_i$'s can be realized as expansive-contractive structure and the remaining $G_i$'s can be realized as expansive-expansive structure, then  the output graph is a nucleation-free graph with implied non-edges.
\end{theorem}

The proof for Theorem \ref{thm:specialmotion} follows directly from the flex-sign proof technique in Section \ref{sec:direct}. Rings of roofs in Section \ref{sec:first} are examples of Theorem \ref{thm:specialmotion}. However, the serious disadvantage of the flex-sign technique is that there is no other known example.

Second, we can apply Henneberg-II constructions on existing graphs to obtain the following scheme:

\begin{scheme}[Henneberg Extender Ring]\label{scheme:HennGeneral}\hfill

{\bf Input graph:} A ring graph $H$ consisting of $G_1$, $G_2$, $\ldots$, $G_n$ such that (1) neighboring subgraphs $G_i$ and $G_{i+1}$ share a hinge $(a_i, b_i)$ with $G_n$ and $G_1$ sharing $(a_n, b_n)$ 
(2) each hinge edge is shared by exactly 2 graphs $G_i$ and $G_{i+1}$ or $G_1$ and $G_n$.

\medskip
{\bf Output graph:} A ring graph $G$ by applying Henneberg-II constructions on $H$ as follows: for $G_i$, add a vertex $v_i$ and four edges $(v_i, a_i)$, $(v_i, b_i)$, $(v_i, a_{i+1})$ and $(v_i, b_{i+1})$. Then remove hinge edges $(a_i, b_i)$.
\end{scheme}

\begin{theorem}\label{thm:HennGeneral}
For the Henneberg Extender Ring scheme, if the following conditions hold:
(1) the input ring graph in independent; (2) each $G_i$ is rigid; and (3) each $G_i$ is nucleation-free after removing its two hinge edges $(a_i, b_i)$ and $(a_{i+1}, b_{i+1)}$, then Scheme \ref{scheme:HennGeneral} outputs a nucleation-free, independent graph with implied non-edges.
\end{theorem}
\begin{proof}
The proof is simple and follows the rank-sandwich proof technique: 
(1) $G$ is nucleation-free: since (a) after removing all hinge edges, $G_i$'s are nucleation-free, and (b) if there is a nucleation, it must include one of the $v_i$'s, but each $v_i$ is incident only to $a_i$, $b_i$, $a_{i+1}$ and $b_{i+1}$, which is not part of any nucleation, since the five vertices do not form a tetrahedron;
(2) Henneberg-II construction always keeps independence of the graph, so $G$ is independent;
(3) a thin-cover argument or a body-hinge argument as in Section \ref{sec:first} can easily show the $(a_i, b_i)$'s are implied.
\end{proof}

\medskip\noindent
{\bf Example of Henneberg extender ring scheme:} The ring of roofs is again an example. As another example, we can use {\em modified octahedral graph}, which is a $K_6$ with $3$ edges missing, as the building blocks of the ring. Fig.~\ref{fig:octahedron} shows one building block of the output of the Henneberg extender ring scheme from a ring of roofs. Note that there are several non-isomorphic $K_6-3$'s and only the one in Fig.~\ref{fig:octahedron} can be obtained using Henneberg extender ring scheme from a ring of roofs.
\begin{center}
\begin{figure}[!htbp]
\begin{center}
\includegraphics[width=.5\textwidth, height=.5\textwidth]{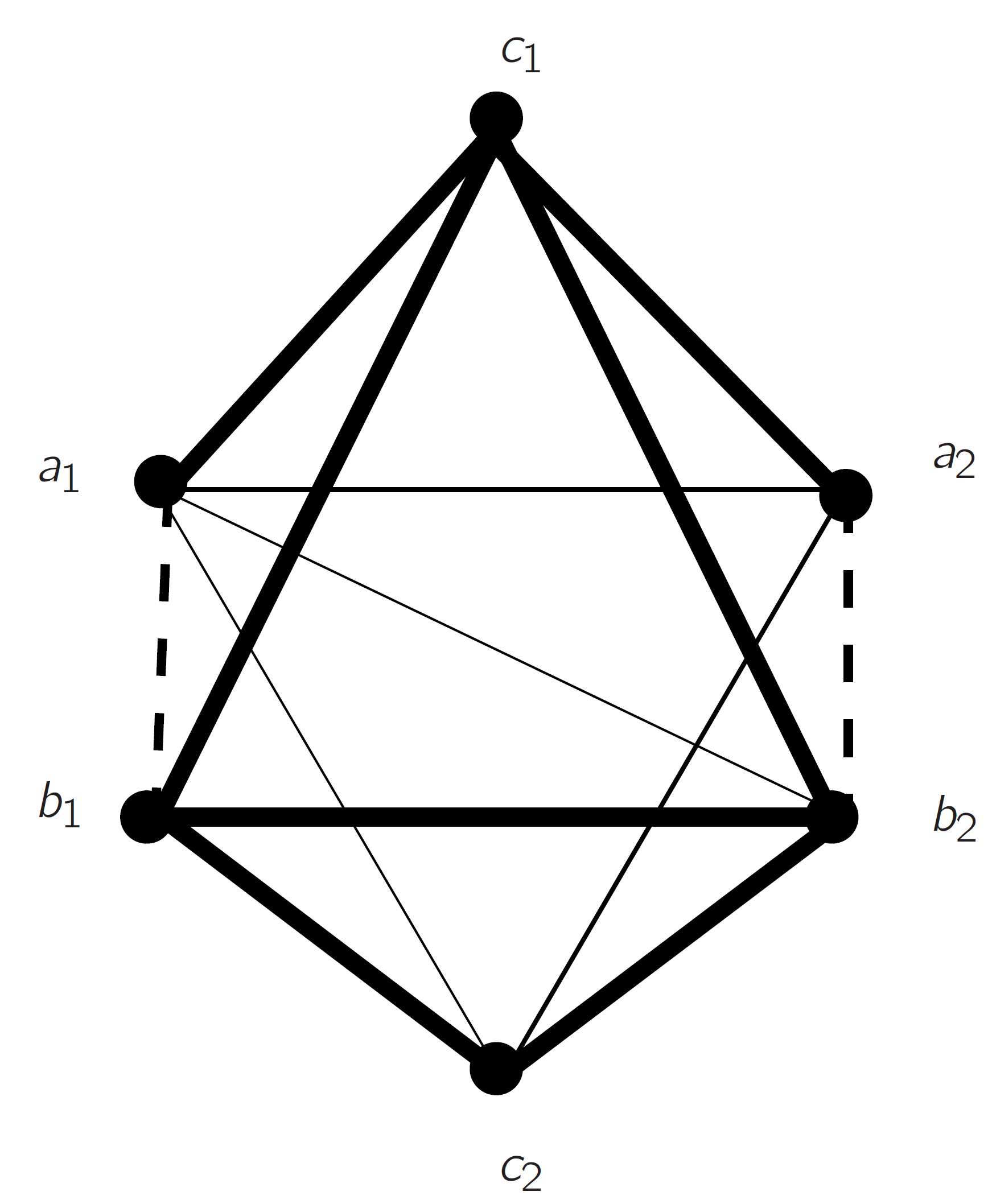}
\end{center}
\caption{A modified octahedral graph obtained from the Henneberg extender ring scheme from a ring of roofs. The dashed lines are the two hinge non-edges.}
\label{fig:octahedron}
\end{figure}
\end{center}

Next we describe another collections of standard schemes, which use $1-$, $2-$ and $3-$sums and standard inductive constructions to build on the existing nucleation-free graphs with implied non-edges. 

\begin{definition}[$k$-sum]
Let $G_1$ and $G_2$ be two graphs that each contains a complete graph on $k$ vertices, $K_k$, as a proper subgraph. For any matching of the vertices of the two $K_k$ ’s by identifying the matched pairs, we can get a new graph $G_3$. We call this procedure a {\em $k$-sum} of $G_1$ and $G_2$ \cite{BelkConnelly07}\cite{Sitharam:2010:CGC}.
\end{definition}

\begin{definition}[Henneberg-I construction]
A {\em Henneberg-I construction} on a graph $G$ in $\mathbb{R}^{3}$ is to first
choose a vertex set $W=\{w_1, w_2$, $w_3\}$ of $3$ vertices of $G$ and then add a new vertex $v$ and three edges $(v,w_1)$, $(v,w_2)$, $(v,w_3)$.
\end{definition}
From \cite{bib:TayWhiteley85}, we know that Henneberg-I construction in $\mathbb{R}^{3}$  preserves independence.

\begin{definition}[vertex split]
Given a graph $G$ and a vertex $u$, incident to vertices $w_1$, $w_2$, $\ldots$, $w_n$, then a {\em vertex split} of $u$ on $i$ edges is a new graph obtained by (1) adding a new vertex $v$, (2) choose $k$ edges $(u, w_1)$, $\ldots$, $(u, w_k)$ incident to $u$ and remove them, then connect $v$ to $w_1$, $w_2$, $\ldots$, $w_k$, and (3) add a new edge $(u, v)$ and $i$ edges from $v$ to $i$ neighbors of $u$. 
\end{definition}

From \cite{WhiteleyVertexSplitting1990}, we know that in $\mathbb{R}^{3}$, vertex split on 0, 1, or 2 edges preserves independence of the graph.

\medskip\noindent
We can use the above concepts to build on existing graphs to obtain another scheme as follows:

\begin{scheme}[Standard-scheme]\label{scheme:standard}\hfill

{\bf Input graphs:} Graphs $G_1$ and $G_2$;

\medskip
{\bf Output graph:} There are four types of output graphs as the following:
\begin{description}
\item[Type I.] A graph $G$ after applying $1-$, $2-$ or $3$-sum on $G_1$ and $G_2$;
\item[Type II.] A graph $G$ after applying Henneberg-I constructions on $G_1$;
\item[Type III.] A graph $G$ after applying Henneberg-II constructions on $G_1$;
\item[Type IV.] A graph $G$ after applying vertex split on 0, 1 or 2 edges on $G_1$.
\end{description}
\end{scheme}

\medskip\noindent
The next theorem proves the correctness of the Standard-scheme.
\begin{theorem}\label{thm:sums}
For the standard-scheme (Scheme \ref{scheme:standard}), we have the following:
\begin{description}
\item[Type I.] If graphs $G_1$ and $G_2$ are both nucleation-free, independent with implied non-edges, then their 1-sum and 2-sum are nucleation-free, independent with implied non-edges. When applying $3$-sum, if the $K_3$ in the intersection of $G_1$ and $G_2$ is not a part of a $K_4$ in either $G_1$ or $G_2$, then the output graph $G$ also nucleation-free, independent with implied non-edges.

\item[Type II.] Let $G_1$ be a nucleation-free, independent with implied non-edges. Apply Henneberg-I construction by adding a new vertex $v$ and three edges $(v,w_1)$, $(v,w_2)$, $(v,w_3)$, such that $w_1$, $w_2$ and $w_3$ is not part of a $K_4$ in $G_1$. Then the output graph $G$ is a nucleation-free, independent with implied non-edges.

\item[Type III.] Let $G_1$ be a nucleation-free, independent. Apply Henneberg-II construction as follows: choose a vertex set $W=\{w_1, w_2$, $w_3, w_4\}$ $G_1$ such that (1) there is at least $1$ edge induced by $W$ and (2) $W$ does not induce a $K_4$, then we delete $1$ edge between the vertices of $W$, and add a new vertex $v$ and four edges $(v,w_1)$, $(v,w_2)$, $(v,w_3)$, $(v,w_4)$. Then the output graph $G$ is nucleation-free and independent.

\item[Type IV.] If $G_1$ is nucleation-free, independent, then applying vertex split on 0, 1, or 2 edges on $G_1$ outputs a nucleation-free and independent graph $G$. 
\end{description}

\end{theorem}

\begin{proof}(Of Theorem \ref{thm:sums})
We prove the three types of output graphs for Scheme \ref{scheme:standard} as the following:

{\bf Type I.} If $G_1$ and $G_2$ are both independent, then (1) it follows by inspecting the motion space of $G$ (right null space of the generic rigidity matrix of $G$) as a combination of the motion spaces of $G_1$ and $G_2$ that the 1-sum and 2-sum of $G_1$ and $G_2$ are both independent; and (2) it follows from the properties of abstract rigidity matroid \cite{graver:servatius:rigidityBook:1993} (especially axiom $C6$) that the 3-sum of $G_1$ and $G_2$ is independent. Moreover, it is easy to see that the implied non-edges in $G_1$ and $G_2$ remain implied in the resulting graphs. For 1-sum and 2-sum, the resulting graph cannot contain any nucleation by simply counting the necessary number of edges for a graph to be rigid. If a 3-sum creates a nucleation, then there must be two subgraphs $G_1^\prime$ of $G_1$ and $G_2^\prime$ of $G_2$ that are rigid and both properly contain the $K_3$ shared by $G_1$ and $G_2$. The only rigid subgraphs available in nucleation-free, independent graphs are trivial ones and in this case, only $K_4$. Thus if the shared $K_3$ is not a part of a $K_4$ in $G_1$ or not a part of a $K_4$ in $G_2$, the resulting graph is nucleation-free.

\medskip
{\bf Type II.} If $G_1$ is nucleation-free, independent with implied non-edges, then after applying Henneberg-I construction, the output graph $G$ is independent since Henneberg-I construction preserves independence. To show $G$ is nucleation-free, we only need to show that the added vertex $v$ is not part of any nucleation. Suppose otherwise, then $v$ is either in (1) a $K_5$ with 1 edge missing or (2) a nucleation with at least $6$ vertices. For (1) to happen, we need $w_1$, $w_2$ and $w_3$ to be part of a $K_4$, which is false. If (2) is true, then we know that in $G_1$, there was a nucleation with at least $5$ vertices, contradiction.

\medskip
{\bf Type III.} If $G_1$ is nucleation-free, independent, then after applying Henneberg-II construction, the output graph $G$ is independent since Henneberg-II construction preserves independence. To show $G$ is nucleation-free, we only need to show that the added vertex $v$ is not part of any nucleation. Suppose otherwise, then $v$ is either in (1) a $K_5$ with 1 edge missing or (2) a nucleation with at least $6$ vertices. For (1) to happen, we need $W$ to induce a $K_4$, which is false. If (2) is true, then we know that in $G_1$, there was a nucleation with at least $5$ vertices, contradiction.

\medskip
{\bf Type IV.} If $G_1$ is nucleation-free, independent, then after applying vertex split on 0, 1, or 2 edges on $G_1$, we have an independent graph $G$. This graph is also nucleation-free since the added vertex connects only to vertices that were incident to a common vertex of $G_1$.
\end{proof} 

\medskip\noindent
{\bf Remark:} The proof of Theorem \ref{thm:sums} (Type I and II) follows from the rank-sandwich proof technique but the rank upper bound ingredient is not needed since implied non-edges are inherited from the constituent graphs. With Theorem \ref{thm:sums}, we can inductively apply $1$-sums and $2$-sums and the resulting graphs are always nucleation-free. But for $3$-sums, if both $K_3$'s we identify were part of a $K_4$ in the original graphs, the resulting graph has a nucleation. Moreover, if we do $3$-sums on nucleation-free rigid graphs, then we can still obtain independent graphs, but those graphs are not nucleation-free.

\section{Roof-addition: General Inductive Construction for nucleation-free, independent Graph with Implied Non-edges}\label{sec:indep}

In this section, we introduce a powerful, new general inductive construction scheme, called {\em roof-addition}, for inductively constructing a variety
of arbitrarily large nucleation-free, independent graphs with implied non-edges. These graphs cannot be constructed using the schemes so far. However, we show that the schemes in the previous section preserve the
starting graph properties needed to apply the roof-addition scheme of this
section. Thus the new, roof-addition scheme can be freely combined with
the schemes in the previous section.

First, Theorem \ref{thm:inductive} shows that a new inductive construction gives independent graphs. The other two properties (nucleation-free and presence of implied non-edges) impose further requirements on the input graphs for the roof-addition Scheme \ref{scheme:roof}. The proof of Theorem \ref{thm:secondInductiveApproach} uses the rank-sandwich technique to show that the roof-addition construction on the appropriate starting graphs result in nucleation-free graphs with implied non-edges. 

\subsection{Inductive construction for independent graphs}
In this section, we introduce an inductive construction for independent graphs.

\begin{center}
\begin{figure}[!htbp]
\includegraphics[width=0.5\textwidth]{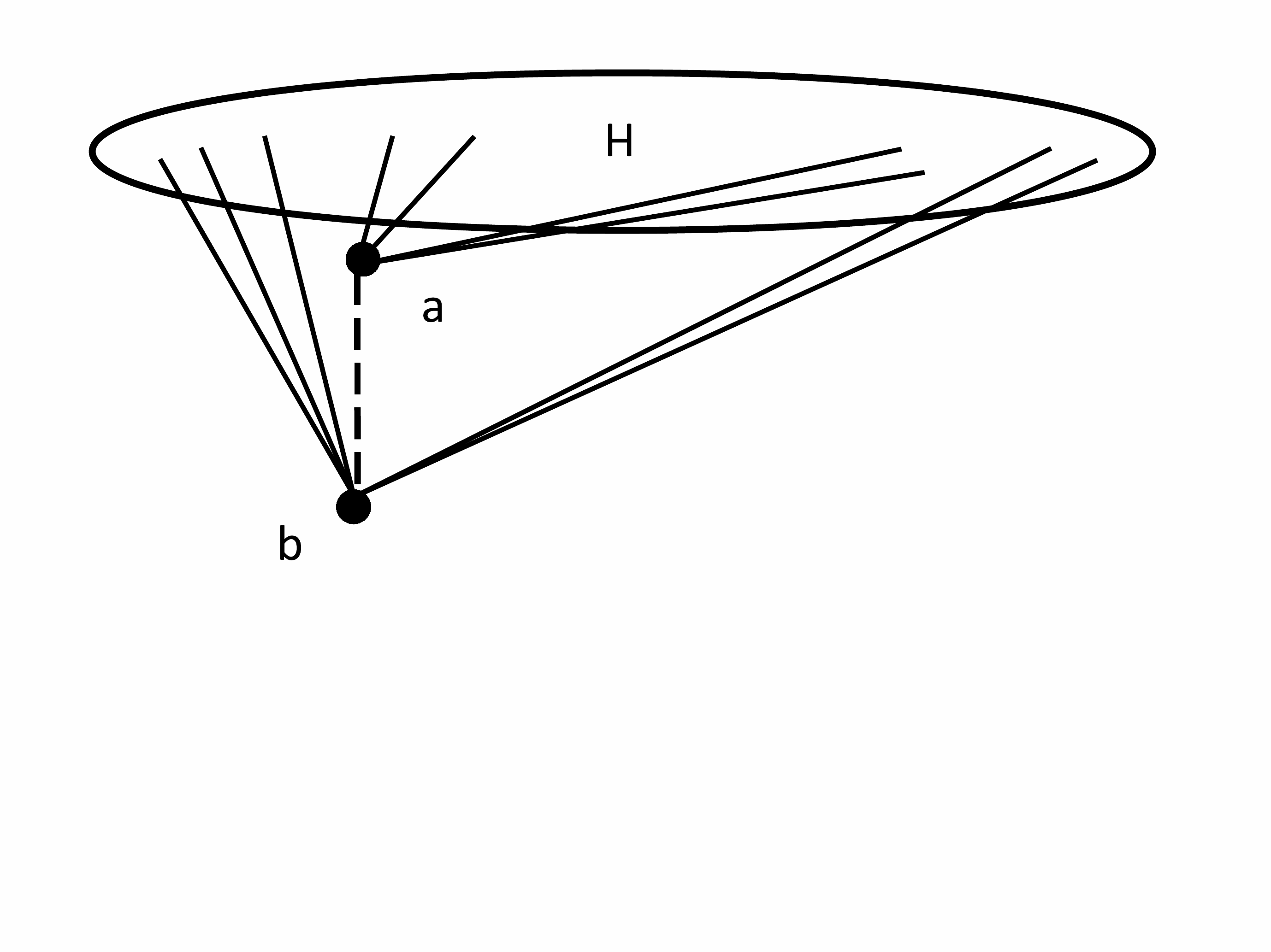}\includegraphics[width=0.5\textwidth]{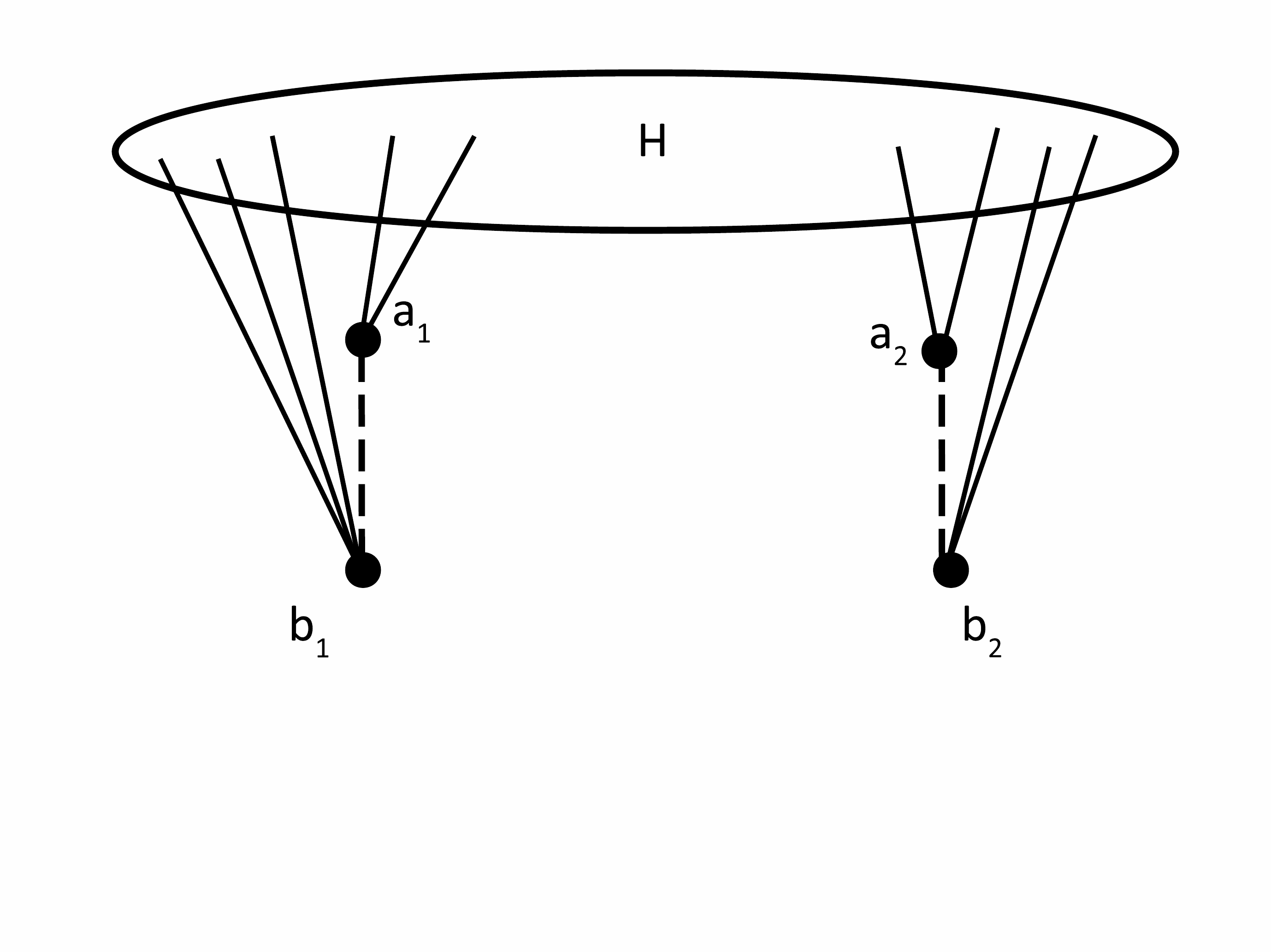}
\caption{A schematic showing how to apply graph cutting of $G$ on $\{a, b\}$}\label{fig:cutting}
\end{figure}
\end{center}

\begin{scheme}[Roof-addition]\label{scheme:roof}\hfill

{\bf Input graph:} Graph $H$ with at least one non-edge.

\medskip
{\bf Output graph:}
A new graph $G$ obtained in $2$ steps.
\begin{description}
\item[Step 1] ({\em graph cutting} along non-edge $\{a, b\}$). Split $a$ into two vertices $a_1$ and $a_2$ and split $b$ into two vertices $b_1$ and $b_2$. Distribute edges of $H$ incident to $a$ by assigning them to $a_1$ and $a_2$ in an arbitrary manner. Distribute edges of $H$ incident to $b$ by assigning them to $b_1$ and $b_2$ in an arbitrary manner. Fig. ~\ref{fig:cutting} shows the procedure of graph cutting.
\item[Step 2]({\em roof pasting}). Take two roofs $R_1$ and $R_2$ sharing a hinge non-edge $\{u, v\}$. The other hinge non-edge of $R_1$ is identified with the vertices $a_1$ and $b_1$, and the other hinge non-edge of $R_2$ is identified with the vertices $a_2$ and $b_2$. Denote the non-hinge vertices of $R_1$ and $R_2$ as $c$ and $c^\prime$ respectively. 
\end{description}
Fig. ~\ref{fig:inductive} shows how to apply roof-addition scheme on a given graph $H$.
\end{scheme}

\begin{center}
\begin{figure}[!htbp]
\includegraphics[width=0.5\textwidth]{roof-additionBeforeCut-eps-converted-to.pdf}\includegraphics[width=0.5\textwidth]{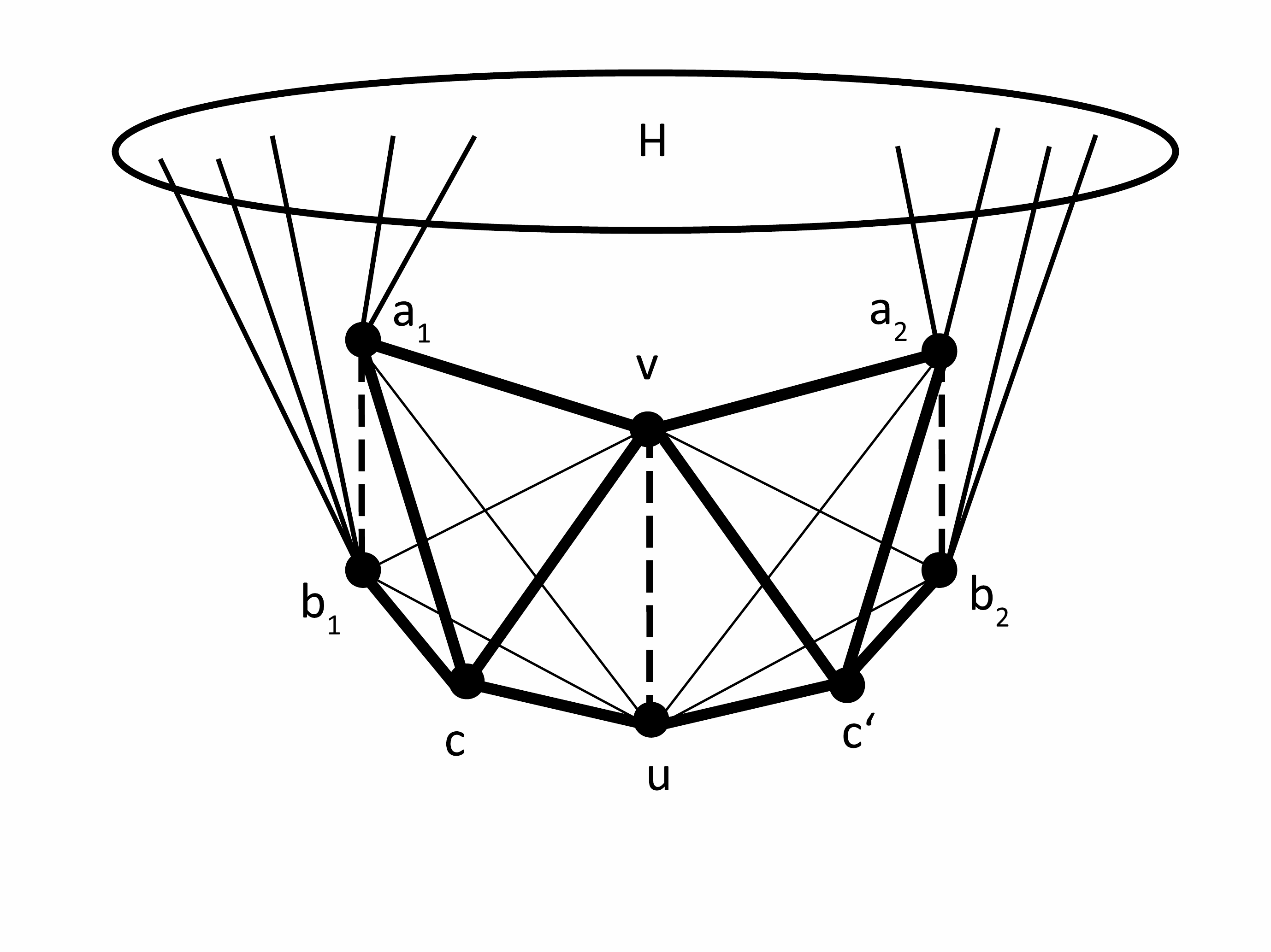}
\caption{A schematic showing how to inductively construct independent graphs by roof-addition. First we identify a non-edge pair $\{a, b\}$, cutting it by splitting $a$ and $b$ and distribute incident edges among each split vertex and its counterpart in an arbitrary fashion. Then we regard two new non-edges as two hinges and add two roofs between them.}\label{fig:inductive}
\end{figure}
\end{center}

Next we show that the roof-addition scheme can be applied to inductively construct independent graphs.
\begin{theorem}[Roof-addition gives independent graphs]\label{thm:inductive}
If $H$ is independent, then the roof-addition scheme (Scheme \ref{scheme:roof}) outputs an independent graph $G$.
\end{theorem}

Before giving the proof of Theorem \ref{thm:inductive}, we first notice that roof-addition is different from existing inductive methods such as Henneberg construction and vertex split that generate independent graphs.
\begin{obs}\label{obs:different}
Given graph $H$, then the graph $G$ generated by roof-addition on $H$, cannot be generated by any combination of Henneberg constructions and/or vertex split on $H$ with edge removal in the end.
\end{obs}
\begin{proof}
Let $\{a, b\}$ be the non-edge pair of $H$ on which the two roofs are added to obtain $G$. We know the edges on $a$ and $b$ are redistributed in order to obtain $G$.

First, applying only Henneberg-I constructions on $H$ does not change the incidence of the original edges of $H$, thus this cannot generate $G$.

If we only apply Henneberg-II construction on $H$, since $\{a, b\}$ is a non-edge of $H$,
we need to select two other vertices $w, x$ $\neq$ $a, b$ on which an incident edge is removed. Applying Henneberg-II constructions can never add back this edge $(w, x)$ and thus the resulting graph will be different from the output graph $G$ of Scheme \ref{scheme:roof}, which has the edge $(w, x)$.

If we only apply vertex split on $H$, there are two cases: (1) we apply a vertex split on a vertex $w$ $\neq$ $a, b$. This will leave the newly added vertex $w^\prime$ multiple paths to $V(H)\setminus\{a, b\}$ without passing $a$ or $b$, where $V(H)$ denotes the vertex set of $H$. This generally is not the same as the output graph of Scheme \ref{scheme:roof}. (2) We apply a vertex split on $a$ or $b$. This means one neighbor $t$ of $a$ will be connected to the newly added vertex $x$. Since vertex split will also add an edge between $a_1$ and $a_2$, in order the make the resulting graph the same as the required output graph, we need to apply edge removal on $(a_1, a_2)$. However, removing $(a_1, a_2)$ means the edge $(w, x)$ will not be removed, thus the resulting graph is different from the output graph of Scheme \ref{scheme:roof}.

When combining the above three constructions, it is obvious that the first step is a vertex split. The case (1) above is still generally different from our Scheme \ref{scheme:roof} due to the reason mentioned above. For case (2), even if we first apply vertex split of $a$ and $b$ on zero edges and manage to use Henneberg-II construction to add $u$ and $v$, the remaining two vertices $c$ and $c^\prime$ cannot be added using any combination of Henneberg constructions and vertex split.
\end{proof}

\medskip\noindent
Note however that known inductive methods, such as vertex split and/or Henneberg constructions, may be able to obtain $G$ with a different independent input graph $H^\prime$. For example, although a ring of roofs $\R$ cannot be obtained using Henneberg-II constructions from a ring of $k-2$ roofs, $\R$ can be obtained using Henneberg-II constructions from a ring of $k$ tetrahedra.

In the remainder of this section, we give the proof of Theorem \ref{thm:inductive}. First, we need the following definition:
\begin{definition}
Let $G=(V, E)$ be a graph with $n$ vertices $\{v_1, \ldots, v_n\}$. Any set of scalars $s_{i,j}=s_{j,i}$ defined for each edge $(v_i, v_j)$ in $E$ is called a {\em stress} for $G$. Moreover, we say $\mathbf{s}=(\ldots, s_{i,j}, \ldots)$ is a {\em self-stress vector} for framework $G(\p)$ if for any vertex point $\textbf{p}_i$ of $G(\p)$, the following {\em stress balance vector equation} at $\textbf{p}_i$ holds:
\begin{equation*}
\sum_j s_{i,j} (\hbox{\textbf{p}}_i -\hbox{\textbf{p}}_j) =0
\end{equation*}
I.e., $\mathbf{s}R(\p)=0$, where $R(\p)$ is the rigidity matrix of $G(\p)$. Each $s_{i,j}$ is called a {\em scalar self-stress} associated with the bar $ij$.
\end{definition}

Note that the stress balance equation at $\textbf{p}_i$ (boldface represents the coordinate position of a vertex) is used to refer to a system of $d$ scalar equations, one for each coordinate of the point $\textbf{p}_i$. We sometimes restrict ourselves to equations corresponding to some specified subset or subspace of the coordinate basis at $\textbf{p}_i$. We refer to this as a specific {\em projection of the stress balance equation} at $\textbf{p}_j$. In addition, we sometimes restrict ourselves to {\em terms $s_{i,j} (\textbf{p}_i -\textbf{p}_j)$ of a specific subset} of the stress balance equation at $\textbf{p}_i$, corresponding to the bars between $\p_i$ and a subset of points, $\textbf{p}_j$.

\begin{proof}
(of Theorem \ref{thm:inductive}) 

We will construct a specific framework of $G$ from any given generic framework $H(\p)$ of $H$ as specified below and shown in Fig.~\ref{fig:ringStress}. 

\begin{itemize}
\item we superpose $\a_2$ on top of $\a_1$ and $\b_2$ on top of $\b_1$. We assume without loss that the hinges are parallel to the $y$ axis. We call the other two hinge points $\A$ and $\B$.
\item the four hinge points  $\A$, $\B$, $\a_1$, $\b_1$ lie on the same plane, without loss, the $xy$-plane, and $\A$, $\B$, $\a_1$, $\b_1$ form a square.
\item the remaining two points $\c$ and $\c^\prime$ lie on a line perpendicular to and passing through the center of the square formed by $\A$, $\B$, $\a_1$, $\b_1$. Notice that this line is parallel to the $z$ axis. 
\end{itemize}

Assume that the new framework is dependent, i.e., there exists a non-zero self-stress vector on edges of the framework such that for each point $\p_i$, the stress balance equation at $\p_i$ holds. We will show that by simply restricting this self-stress vector to the edges of $H$, we get a non-zero self-stress vector on a generic framework $H(\p)$ of $H$ that is obtained by gluing together or identifying the points $\a_1$ and $\a_2$, and similarly the points $\b_1$ and $\b_2$.
But we knew $H$ is independent. Thus we draw a contradiction, thereby proving the theorem.

\begin{center}
\begin{figure}[!htbp]
\centering
\scalebox{0.5}[0.5]{
\includegraphics[width=\textwidth]{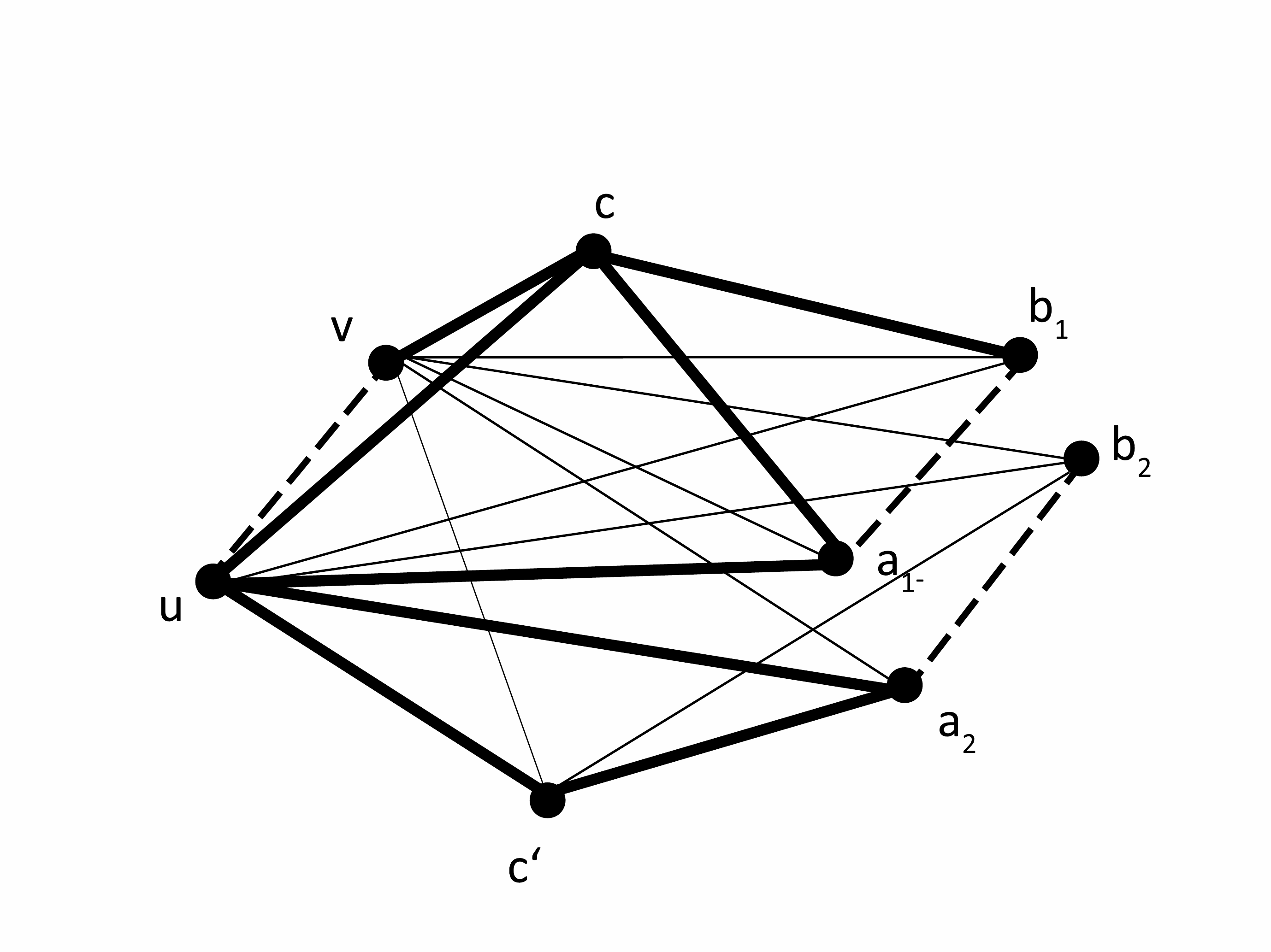}}
\caption{A schematic showing the frameworks, i.e., including the positioning of the vertices for the last two roofs for the proof of Theorem \ref{thm:inductive}. The two roofs are viewed different from a different perspective than Fig.~\ref{fig:inductive}. Note that the two points $\a_1$ and $\a_2$ represent different vertices that are coincident and similarly $\b_1$ and $\b_2$ represent different vertices that are coincident.}\label{fig:ringStress}

\end{figure}
\end{center}

\noindent
A closer look at the stress balance equations gives us the following claims.

\begin{clm}\label{clm:stressC}
The projection of $s_{\c,\A}(\c-\A)$ on the plane $xoy$ is equal to the projection of $s_{\c,\b_1}(\c-\b_1)$ on the plane $xoy$. The same happens between $s_{\c^\prime,\B}(\c^\prime-\B)$ and $s_{\c^\prime,\a_2}(\c^\prime-\a_2)$. The remaining two stresses on $\c$ have the same magnitude and so do the remaining stresses on $\c^\prime$.
\end{clm}
\begin{proof}
Consider all stress equations terms at $\c$. All those four terms add up to zero. In particular, if we project the terms on the plane $xoy$, then their projections lie on two lines that are perpendicular to each other. Thus sum of the projections on each line should add up to zero, thus proving the claim. See Fig.~\ref{fig:ringStressC}.
\end{proof} 

\begin{center}
\begin{figure}[!htbp]
\centering
\scalebox{0.5}[0.5]{
\includegraphics[width=\textwidth]{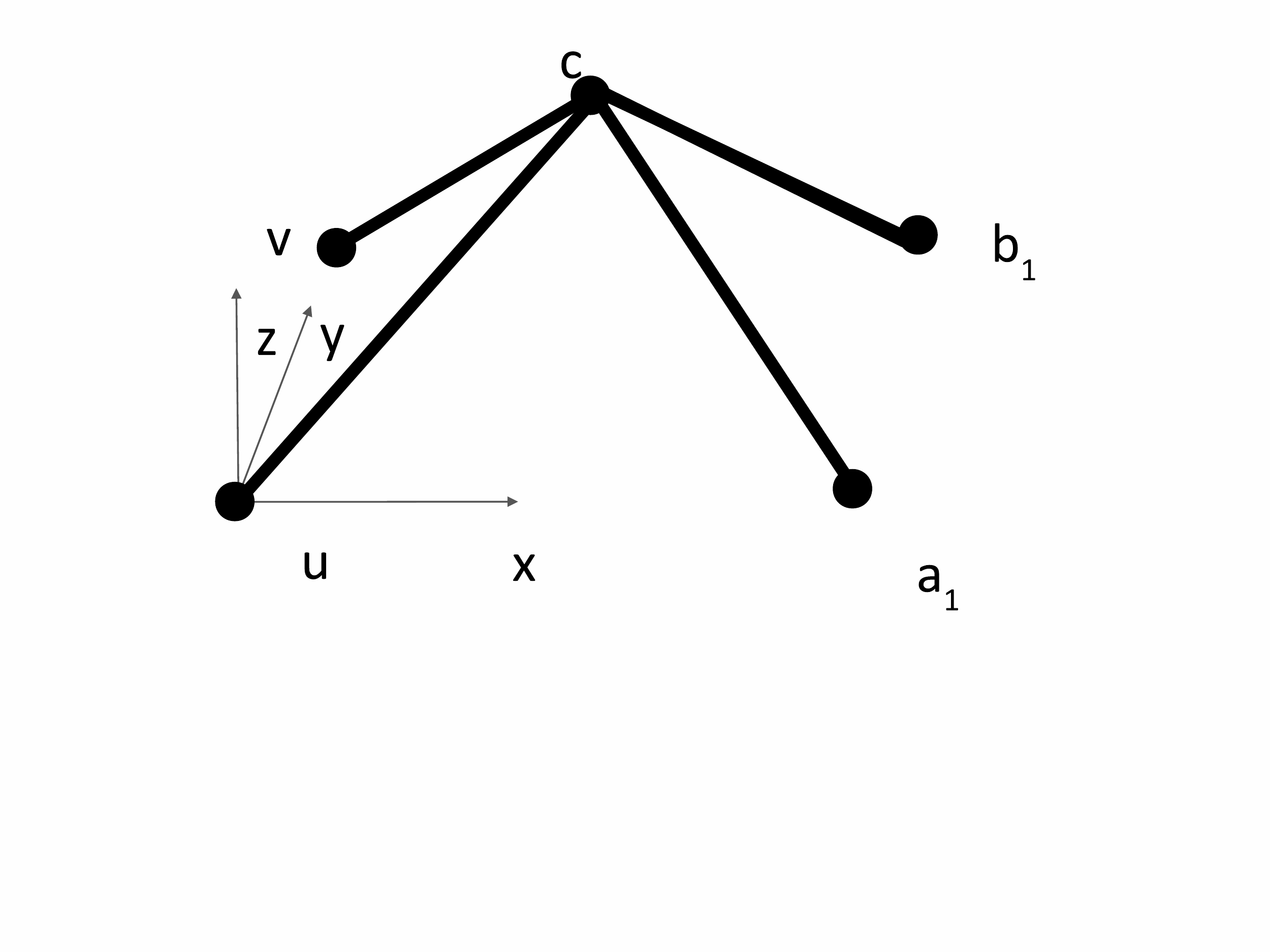}}
\caption{A figure showing the symmetry of the five edges at $\c$.}\label{fig:ringStressC}

\end{figure}
\end{center}

\begin{clm}\label{clm:stressA}
Consider the terms of the stress balance equation at $\a_1$ and $\a_2$ corresponding to the bars in the two added roofs. All these terms together add up to zero. Similarly, the analogous terms at $\b_1$ and $\b_2$ add up to zero.
\end{clm}
\begin{proof}
We will show the part about $\a_1$ and $\a_2$ first.

There are three terms for the stress balance equation at $\a_1$ restricted to the added graph, and three terms for the stress balance equation at $\a_2$ restricted to the $(n+2)^{nd}$ roof as well. We will show that four of these six terms add up to zero i.e.,
\begin{eqnarray}
&&s_{\a_2,\c^\prime} ({\a_2-\c^\prime})   +   s_{\a_1,\c} (\a_1-\c) \notag \\& +&s_{\a_1,\B} (\a_1-\B)   +   s_{\a_2,\B} (\a_2-\B) \label{eqn:fourterm}
\\ &=&0 \notag
\end{eqnarray}
and the remaining two terms also add up to zero,

\begin{equation}\label{eqn:twoterm}
s_{\a_2,\A} (\a_2-\A)   +   s_{\a_1,\A} (\a_1-\A)   =0 .
\end{equation}

To show \eqref{eqn:fourterm}, we need to inspect the six stress equation terms at $\B$, especially the four terms correspond to edges that are incident to $\c,\c^\prime$, $\a_1$ and $\a_2$.
Reflection symmetry of the octahedron and Claim \ref{clm:stressC} together show that the six terms in
equation \eqref{eqn:fourterm}- i.e., three terms in the stress balance equation at $\a_2$ and three terms in the stress balance equation at $\a_1$ -  are obtained as the reflection of four corresponding terms of the stress balance equation at $\B$.
Thus if we can show that the latter $4$ terms add up to zero,
then the former $6$ terms add up to zero as well, thereby showing Equation \eqref{eqn:fourterm}.

\begin{center}
\begin{figure}[!htbp]
\centering
\scalebox{0.5}[0.5]{
\includegraphics[width=\textwidth]{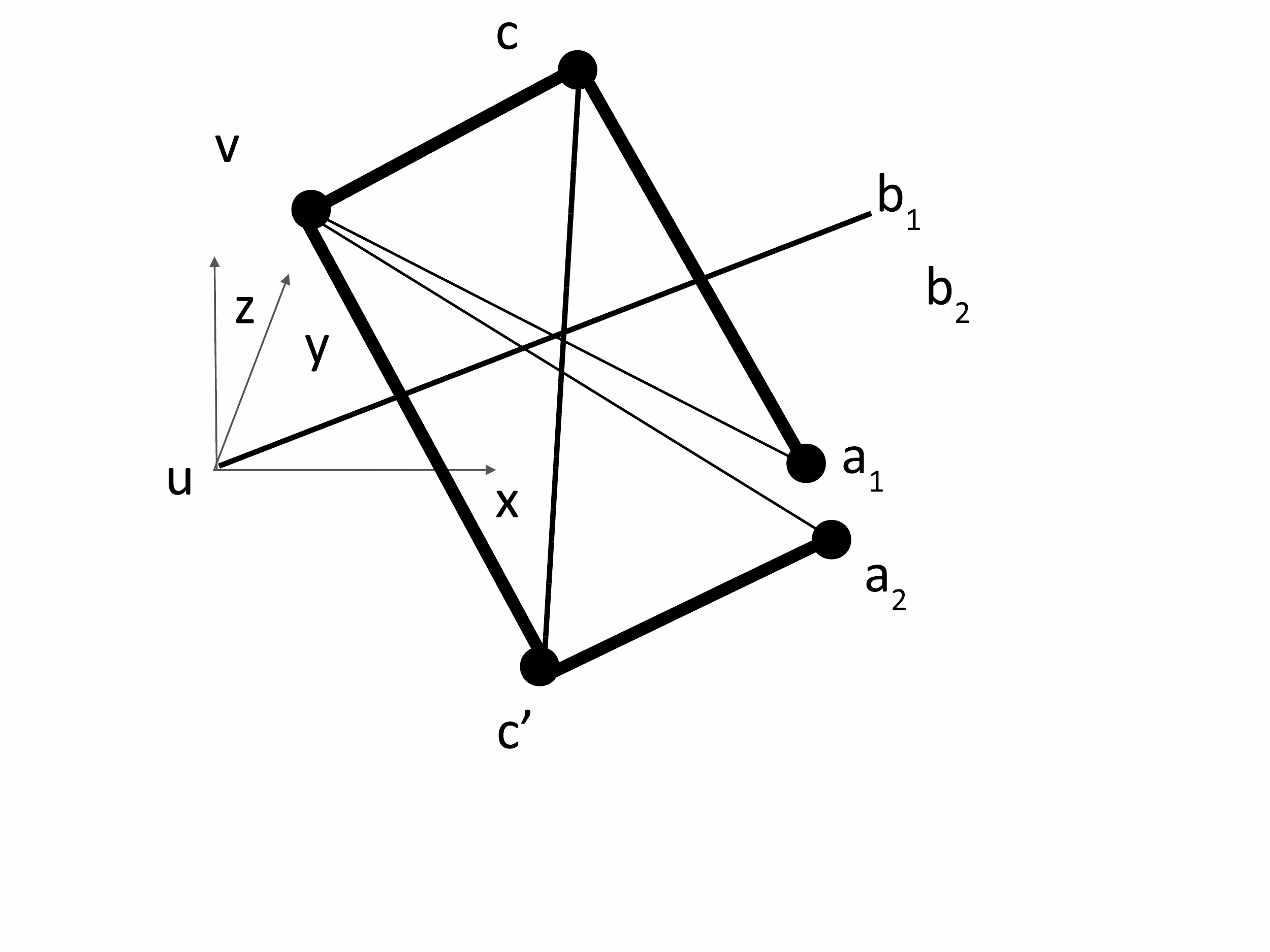}
}
\caption{A figure showing the four of the six terms that add up to zero at $\B$ and their reflection through the plane passing $\c\c^\prime$ and $\A\a_2$ (the two dashed line) to the four terms in \eqref{eqn:fourterm}.}\label{fig:ringStressReflection}

\end{figure}
\end{center}

More precisely, denote by $\mathbf{s}^\prime_{\B}$ the stress vector restricted to these specific $4$
edges at $\B$. Denote by $R^\prime_{\B}$ the corresponding $4$ by $3$
submatrix of the rigidity matrix (whose row vectors specify the corresponding $4$ bars at $\B$). Their product $\mathbf{s}^\prime_{\B}R^\prime_{\B}$ is a vector whose entries are the $4$ terms at $\B$.
Consider the plane containing the two lines $\c\c^\prime$ and $\A\b_2$, i.e., the dashed lines shown in Figure \ref{fig:ringStressReflection}. Denote by $L$ the reflection across this plane. Now notice that the entries of the product $\mathbf{s}^\prime_{\B}R^\prime_{\B}L$ are precisely the $4$ terms in Equation \eqref{eqn:fourterm}.

The proof about $\b_1$ and $\b_2$ is essentially the same as above, the only difference is that in \eqref{eqn:fourterm}, we have only four terms instead of six. We can still use the reflection symmetry and obtain the exact same result.
\end{proof}

\medskip\noindent
Thus we know the stress balance equation terms at $\a_1$ restricted to $H$ and $\a_2$ restricted to $H$ add up to zero. The same result can be obtained at $\b_1$ and $\b_2$. So if we remove the two added roofs and use the same stress balance equation on each edge, we know that the stress balance equation terms add up to zero on every point of the generic framework of $H$. Thus the restricted self-stress vector of the new framework is a self-stress vector on a generic framework of $H$. Next we will show this self-stress vector is non-zero.

\begin{clm}\label{clm:nonzero}
Suppose $\mathbf{s}$ is a non-zero self-stress vector for $G$ (which has the two added roofs). Denote by $\mathbf{s}^\prime$ self-stress vector $\mathbf{s}$ restricted to $H$, after we remove the last two roofs and glue $\a_2$ with $\a_1$ and $\b_2$ with $\b_1$. Then $\mathbf{s}^\prime$ is a non-zero self-stress vector for $H(\p)$.
 \end{clm}
 \begin{proof}
We argue the existence of non-zero stress in the original framework for all the following cases:
\begin{enumerate}[i]
\item\label{case:cb} If there is a non-zero stress at $(\c, \b_1)$ (resp. $(\c, \a_1)$, $(\c^\prime, \b_2)$, $(\c^\prime, \a_2)$), then the stress equation terms of the two added roofs will always have a $z$ projection at $\b_1$ (resp. $\a_1$, $\b_2$, $\a_2$). So the original framework $H(\p)$ must have a non-zero stress to cancel out this $z$ projection. When there is no edge in $H$ that is incident to $\b_1$ (resp. $\a_1$, $\b_2$, $\a_2$), we know the stress at $(\c, \b_1)$ (resp. $(\c, \a_1)$, $(\c^\prime, \b_2)$, $(\c^\prime, \a_2)$) has to be zero.

\item\label{case:cA} If there is a non-zero stress at $(\c, \A)$ (resp. $(\c, \B)$, $(\c^\prime, \A)$, $(\c^\prime, \B)$), then from Claim \ref{clm:stressC}, we know there is a non-zero stress at $(\c, \b_1)$ (resp. $(\c, \a_1)$, $(\c^\prime, \b_2)$, $(\c^\prime, \a_2)$), which is equal to Case \ref{case:cb}.

\item\label{case:Ab} If there is a non-zero stress at $(\A, \b_1)$ (resp. $(\B, \a_1)$, $(\A, \b_2)$, $(\B, \a_2)$), then either (1) there is a non-zero stress at $(\c, \b_1)$ (resp. $(\c, \a_1)$, $(\c^\prime, \b_2)$, $(\c^\prime, \a_2)$), which is Case \ref{case:cb}, or (2) the stress at $(\c, \b_1)$ (resp. $(\c, \a_1)$, $(\c^\prime, \b_2)$, $(\c^\prime, \a_2)$) is zero. When (2) happens, the stress equation terms of the two added roofs will always have a $y$ projection at $\b_1$ (resp. $\a_1$, $\b_2$, $\a_2$). So the original framework must have a non-zero stress to cancel out this $y$ projection. Or if there is no edge in $H$ that is incident to $\b_1$ (resp. $\a_1$, $\b_2$, $\a_2$), then the stress at $(\c, \b_1)$ (resp. $(\c, \a_1)$, $(\c^\prime, \b_2)$, $(\c^\prime, \a_2)$) has to be zero. 

\item\label{case:Bb} If there is a non-zero stress at $(\B, \b_1)$ (resp. $(\A, \a_1)$, $(\B, \b_2)$, $(\A, \a_2)$), then there are 3 further cases. (1) There is a non-zero stress at $(\c, \b_1)$ (resp. $(\c, \a_1)$, $(\c^\prime, \b_2)$, $(\c^\prime, \a_2)$), which is Case \ref{case:cb}. (2) The stress at $(\c, \b_1)$ (resp. $(\c, \a_1)$, $(\c^\prime, \b_2)$, $(\c^\prime, \a_2)$) is zero, there is a non-zero stress at $(\A, \b_1)$ (resp. $(\B, \a_1)$, $(\A, \b_2)$, $(\B, \a_2)$). This is just Case \ref{case:Ab}. (3) The stresses at $(\c, \b_1)$ (resp. $(\c, \a_1)$, $(\c^\prime, \b_2)$, $(\c^\prime, \a_2)$) and $(\A, \b_1)$ (resp. $(\B, \a_1)$, $(\A, \b_2)$, $(\B, \a_2)$) are both zero. Then the stress equation terms of the two added roofs will always have a $x$ projection at $\b_1$ (resp. $\a_1$, $\b_2$, $\a_2$). So the original framework $H(\p)$ must have a non-zero stress to cancel out this $x$ projection. Or if there is no edge in $H$ that is incident to $\b_1$ (resp. $\a_1$, $\b_2$, $\a_2$), then the stress at $(\c, \b_1)$ (resp. $(\c, \a_1)$, $(\c^\prime, \b_2)$, $(\c^\prime, \a_2)$) has to be zero. 
\end{enumerate}
\end{proof} 
 
 \medskip\noindent
 Returning to the proof of Theorem \ref{thm:inductive}, using Claims \ref{clm:stressC}, \ref{clm:stressA} and \ref{clm:nonzero}, we have found a non-zero self-stress for the original framework of $H$. I.e., the generic framework of $H$ is dependent. Contradiction.
 \end{proof} 

\subsection{Roof-addition gives nucleation-free graphs with implied non-edges}
In this section, we show that the roof-addition scheme gives nucleation-free graphs with implied non-edges in the following theorem.

\begin{theorem}[Roof-addition gives nucleation-free, independent graphs with implied non-edges]\label{thm:secondInductiveApproach}
Let $H$ be an independent graph satisfying the following: there is a 2-thin cover $\mathcal{X}$ $=\{X_1, X_2, \ldots, X_n\}$ of $H$ such that (1) the shared part $(V, S(\mathcal{X}))$ is independent; and  (2) $rank(H)$ $=$ $\sum_{X_i\in \mathcal{X}}$ $rank(H^\star[X_i])$ $-$ $\sum_{\{u, v\}\in S(\mathcal{X})}$ $(d(u, v) - 1)$, where $H^\star$ $=$ $H\cup S(\mathcal{X})$. Let $\{a, b\}$ $\in$ $S(\mathcal{X})$ be a non-edge of $H$ for which the graph cutting operation of Scheme \ref{scheme:roof} is applied to $H$ along $\{a, b\}$ in the following manner. (1) For any covering subgraph $H[X_j]$ of $H$ that has at least one shared non-edge, if one of its edges is assigned to $a_1$ or $b_1$ (resp. $a_2$ or $b_2$), then all of its edges incident to $a$ or $b$ are assigned to $a_1$ or $b_1$ (resp. $a_2$ or $b_2$). (2) After graph cutting, the graph is nucleation-free. Note that such a $\X$ may not exist. When it exists, we call $H$ a {\em starting graph for Theorem \ref{thm:secondInductiveApproach}}. Applying graph cutting on a starting graph $H$ in the above manner results in an output graph $G$ that is nucleation-free, independent with implied non-edges.
\end{theorem}

\begin{proof}(of Theorem \ref{thm:secondInductiveApproach})
The proof uses the rank-sandwich proof technique. First, the independence of $G$ is guaranteed by Theorem \ref{thm:inductive} and it is also clear that $G$ is nucleation-free. 

We will use the rank upper bound argument in Section \ref{ringupper1} to show the existence of implied non-edges in $G$.

In fact, we will show that whenever a graph cutting preserves all covering subgraphs that contain a shared non-edge (it is clear that the graph cutting in Theorem \ref{thm:secondInductiveApproach} is one such cutting), then the rank upper bound argument works. When we say a covering subgraph is preserved, we mean after graph cutting, its corresponding subgraph does not contain both $a_1$ and $a_2$ (or $b_1$ and $b_2$). We can find a cover $\X^\prime$ of $G$ by modifying $\X$. We use $V(G)$ (resp. $V(H)$) to denote the vertex set of graph $G$ (resp. $H$) and $E(G)$ (resp. $E(H)$) to denote the edge set of graph $G$ (resp. $H$). 

We will start from the condition $|E(H)|$ $=$ $rank(H)$ $=$ $\sum_{X_i\in \mathcal{X}}$ $rank(H^\star[X_i])$ $-$ $\sum_{\{u, v\}\in S(\mathcal{X})}$ $(d(u, v) - 1)$ and then obtain a rank IE count on $G^\star$(i.e., $G\cup S(\X^\prime)$) to show that rank$(G^\star)$ is equal to $|E(G)|$. First, we notice that there are three types of covering subgraphs of $H$ in $\X$:
\begin{enumerate}[(i)]
\item trivial covering subgraphs, i.e., edges. We denote this set as $\X_0$.
\item non-trivial covering subgraphs $H[X_i]$ that are preserved in $\X^\prime$ and we denote those covering subgraphs as the set $\X_1$.

\item non-trivial covering subgraphs $H[X_j]$ that are not preserved in $\X^\prime$ and we denote those covering subgraphs as the set $\X_2$. I.e., these covering subgraphs fall apart so that their edges become trivial covering subgraphs in $\X^\prime$, except those edges that are present in some covering subgraph of $\X_1$. We note that these covering subgraphs only have {\em edges} shared by other covering subgraphs in $\X$. 
\end{enumerate}

Let $d_{\X_i}(u, v)$ ($i=1$ or $2$) be the number of sets $X_i$ in $\mathcal{X}_i$ such that $\{u, v\} \subseteq X_i$. Let $S(\X_i)$ ($i=1$ or $2$) be the
set of all pairs of vertices $\{u, v\}$ such that $X_j \cap X_k = \{u, v\}$ for some $X_j\in \X_i$ and $X_k \in \X_i$. Let $L:=S(\X)$ $\setminus$ $S(\X_1)$. Then we can rewrite the rank condition on $H$ as follows:
\begin{eqnarray*}
|E(H)|=\text{rank}(H)& = & \sum_{X_i\in \mathcal{X}} rank(H^\star[X_i]) - \sum_{\{u, v\}\in
S(\mathcal{X})} (d(u, v) - 1)\\
&=& \sum_{X_i\in \mathcal{X}_1} rank(H^\star[X_i]) - \sum_{\{u, v\}\in
S(\mathcal{X}_1)} (d_{\X_1}(u, v) - 1) \\
&& - \sum_{\{u, v\}\in L} (d_{\X_1}(u, v) +d_{\X_2}(u, v) - 1) \\
&&+ \sum_{X_i\in \mathcal{X}_2} rank(H^\star[X_i])
\\&& +\sum_{X_i\in \mathcal{X}_0} 1 
\end{eqnarray*}

It is clear that the only non-trivial covering subgraphs besides the two added roofs in $\X^\prime$ are in $\X_1$. Since in the rank IE count, the contribution of trivial covering subgraphs is equal to their number of edges, we will show:

\begin{eqnarray}
 \sum_{X_i\in \mathcal{X}_2} rank(H^\star[X_i]) +   \sum_{X_i\in \mathcal{X}_0} 1 && \nonumber\\
  -\sum_{\{u, v\}\in L}(d_{\X_1}(u, v) +d_{\X_2}(u, v) - 1) &= &  \sum_{e\in E(H)\setminus E(\X_1)} 1, \label{eqn:goal1}
\end{eqnarray}
where $E(\X_1)$ denotes the set of edges induced by vertex sets of $\X_1$.

First, since $H$ is independent, we know every covering subgraph in $\X_2$ is independent. Thus every edge in $\X_2$ that is not in $L$ contributes $1$ to the to the left hand side of Equation (\ref{eqn:goal1}).

For $\{u, v\}\in L$, we first note that every $\{u, v\}$ is an edge, since the covering subgraphs in $\X_2$ do not contain shared non-edges in $S(\X)$. Another key observation is that for every $\{u, v\}\in L$, (1) $d_{\X_1}(u, v)$ $=$ $1$ when $\{u, v\}$ is in some covering subgraph of $\X_1$, since otherwise $\{u, v\}\in
S(\mathcal{X}_1)$; and (2) $d_{\X_1}(u, v)$ $=$ $0$ when $\{u, v\}$ is not in any covering subgraph of $\X_1$.

When (1) happens, $\{u, v\}$ is an edge in some covering subgraph of $\X_1$ and it is clear that $d_{\X_2}(u, v)$ is equal to the number of $X_i$'s that are in $\X_2$ and contain $\{u, v\}$. I.e., these edges contribute zero to the left hand side of Equation (\ref{eqn:goal1}).

When (2) happens, $\{u, v\}$ is not an edge in any covering subgraph of $\X_1$ and each $\{u, v\}$ contributes $1$ to the left hand side of Equation (\ref{eqn:goal1}), since $d_{\X_1}(u, v)=0$ and every covering subgraph in $\X_2$ is independent. Thus the total contribution of the left hand side of Equation (\ref{eqn:goal1}) is equal to the number of edges of $\X_2$ that do not appear in $\X_1$. Hence we have Equation (\ref{eqn:goal1}) and more importantly, the following:
\begin{eqnarray}
|E(H)|=\text{rank}(H)&=& \sum_{X_i\in \mathcal{X}_1} rank(H^\star[X_i]) - \sum_{\{u, v\}\in
S(\mathcal{X}_1)} (d_{\X_1}(u, v) - 1) \nonumber \\
&&+ \sum_{e\in E(H)\setminus E(\X_1)} 1 \label{eqn:goal2}
\end{eqnarray}

To apply rank upper bound argument on $G$, we need to show that rank$(G)$ $=$ $\sum_{X_i\in \mathcal{X}^\prime}$ $rank(G^\star[X_i])$ $-$ $\sum_{\{u, v\}\in S(\mathcal{X}^\prime)}$ $(d(u, v) - 1)$. Since $G$ is independent, all we need is to show that $|E(G)|$ $=$ $\sum_{X_i\in \mathcal{X}^\prime}$ $rank(G^\star[X_i])$ $-$ $\sum_{\{u, v\}\in S(\mathcal{X}^\prime)}$ $(d(u, v) - 1)$. 

Next we turn to the covering subgraphs of $\X^\prime$. There are two types.
\begin{enumerate}[(i)]
\item trivial covering subgraphs, i.e., edges. We denote this set as $\X_0^\prime$ and we know $|\X_0^\prime|$ $=$ $\sum_{e\in E(H)\setminus E(\X_1)} 1$.
\item non-trivial covering subgraphs $H[X_i]$ $\in$ $\X_1$ that are preserved as $G[X_i]$ in $\X^\prime$ and two added roofs $R_1$ and $R_2$.
\end{enumerate}

It is clear that the number of edges of $G$ can be written as follows:
\begin{eqnarray*}
|E(G)| &=& |E(H)| + \text{rank}(G[R_1]) + \text{rank}(G[R_2])\\
&=&|E(H)| + \text{rank}(G^\star[R_1]) + \text{rank}(G^\star[R_2])-2
\end{eqnarray*}
Plugging in Equation (\ref{eqn:goal2}), we have:
\begin{eqnarray*}
|E(G)| &=&  \text{rank}(G^\star[R_1]) + \text{rank}(G^\star[R_2])-2\\
& +& \sum_{X_i\in \mathcal{X}_1} rank(H^\star[X_i]) - \sum_{\{u, v\}\in
S(\mathcal{X}_1)} (d_{\X_1}(u, v) - 1)+ \sum_{e\in \X_0^\prime} 1 
\end{eqnarray*}
The final step is to build a relationship between $\sum_{X_i\in \mathcal{X}_1} rank(H^\star[X_i])$ $-$ $\sum_{\{u, v\}\in S(\mathcal{X}_1)}$ $(d_{\X_1}(u, v) - 1)$ and $\sum_{X_i\in \mathcal{X}_1} rank(G^\star[X_i])$ $-$ $\sum_{\{u, v\}\in
S(\mathcal{X}^\prime)}$ $(d_{\X^\prime}(u, v) - 1)$. We notice that (1) every covering subgraph in $\X_1$ is preserved in $\X^\prime$, so $\sum_{X_i\in \mathcal{X}_1}$ $ rank(H^\star[X_i])$ $=$ $\sum_{X_i\in \mathcal{X}_1}$ $rank(G^\star[X_i])$; and (2) $\{a, b\}$ is the only shared non-edge that is changed. More precisely, $\{a, b\}$ is split into $\{a_1, b_1\}$ and $\{a_2, b_2\}$. Together with the fact that $\{a_1, b_1\}$ and $\{a_2, b_2\}$ are both shared by the two added roofs, we know $d_{\X_1}(a, b)$ $=$ $d_{\X^\prime}(a_1, b_1)-1$ $+$ $d_{\X^\prime}(a_2, b_2)-1$ $=$ $d_{\X^\prime}(a_1, b_1)$ $+$ $d_{\X^\prime}(a_2, b_2)$ $-2$. Let $\{w, t\}$ be the non-edge shared by $R_1$ and $R_2$. Hence we have:
\begin{eqnarray*}
|E(G)| &=& \text{rank}(G^\star[R_1])+\text{rank}(G^\star[R_2]) -2 \\
&&+ \sum_{X_i\in \mathcal{X}_1} \text{rank}(G^\star[X_i]) - \sum_{\{u, v\}\in
S(\mathcal{X}_1)\setminus \{a, b\}} (d_{\X_1}(u, v) - 1)\\
&& -(d_{\X_1}(a, b) -1)\\
 && + \sum_{e\in \X_0^\prime} 1\\
&=& \text{rank}(G^\star[R_1])+\text{rank}(G^\star[R_2]) -2\\
&&+ \sum_{X_i\in \mathcal{X}_1} \text{rank}(G^\star[X_i]) - \sum_{\{u, v\}\in
S(\mathcal{X}_1)\setminus \{a, b\}} (d_{\X_1}(u, v) - 1)\\
&&-(d_{\X^\prime}(a_1, b_1) + d_{\X^\prime}(a_2, b_2) -2 -1) \\
&& + \sum_{e\in \X_0^\prime} 1 \\
&=&  \text{rank}(G^\star[R_1])+\text{rank}(G^\star[R_2]) -1\\
&&+\sum_{X_i\in \mathcal{X}_1} \text{rank}(G^\star[X_i]) - \sum_{\{u, v\}\in
S(\mathcal{X}_1)\setminus \{a, b\}} (d_{\X_1}(u, v) - 1) \\
&& -(d_{\X^\prime}(a_1, b_1)-1) -(d_{\X^\prime}(a_2, b_2) -1) \\
&& + \sum_{e\in \X_0^\prime} 1 \\
&=&  \text{rank}(G^\star[R_1])+\text{rank}(G^\star[R_2]) - (d(w, t)-1)\\
&&+\sum_{X_i\in \mathcal{X}_1} \text{rank}(G^\star[X_i]) - \sum_{\{u, v\}\in
S(\mathcal{X}_1)\setminus \{a, b\}} (d_{\X_1}(u, v) - 1) \\
&& -(d_{\X^\prime}(a_1, b_1)-1) -(d_{\X^\prime}(a_2, b_2) -1) \\
&& + \sum_{e\in \X_0^\prime} 1 \\
&=& \sum_{X_i\in \mathcal{X}^\prime} \text{rank}(G^\star[X_i]) - \sum_{\{u, v\}\in
S(\mathcal{X}^\prime)} (d_{\X^\prime}(u, v) - 1). 
\end{eqnarray*}
Noticing that $(V(H), S(\X))$ is independent, we know that $(V(G), S(\X^\prime))$ is independent, since the only differences between $S(\X)$ and $S(\X^\prime)$ are the split of $\{a, b\}$ and the addition of the shared non-edge of $R_1$ and $R_2$, which is disjoint from the rest of $S(\X^\prime)$. Thus we can apply Theorem \ref{IPindep} to obtain that rank$(G^\star)$ $\leq$ $\sum_{X_i\in \mathcal{X}^\prime} rank(G^\star[X_i])$ $-$ $\sum_{\{u, v\}\in
S(\mathcal{X}^\prime)}$ $(d_{\X^\prime}(u, v) - 1)$ $=$ $|E(G)|$ $=$ rank$(G)$. Since $G\subseteq G^\star$ and $G$ is independent, we know rank$(G^\star)$ $\ge$ rank$(G)$ and hence rank$(G^\star)$ $=$ rank$(G)$. It follows that every non-edge in $S(\mathcal{X}^\prime)$ is implied. These include (1) all non-edges in $S(\X)$ except $\{a, b\}$; (2) $\{a_1, b_1\}$ and $\{a_2, b_2\}$; and (3) the non-edge $\{w, t\}$ shared by the two added roofs $R_1$ and $R_2$.
\end{proof} 



\subsection{Existence and generation of starting graphs for Theorem \ref{thm:secondInductiveApproach}}\label{sec:indepgeneral}
In this section, we give several concrete example starting graphs for Theorem \ref{thm:secondInductiveApproach}. We show the existence of several fixed-size base starting graphs that cannot be generated by the Schemes \ref{schme:specialmotion}-\ref{scheme:roof}. Then, we show that the schemes given so far can be inductively and freely combined to generate starting graphs, and
thereby generate a large variety of arbitrarily large, nucleation-free, independent graphs with implied non-edges. Note that for fixed-size base starting graphs, their independence is verified by symbolically computing the
generic rank of a rigidity matrix of indeterminates for frameworks
corresponding to those example starting graphs. Additionally, we are not relying on
special positions of the joints(vertices) and we are showing that the rank
is maximal, i.e., equal to the number of edges, so in any case, we do not
have to worry about numerical rounding issues distorting the rank
computation since such distortions can only decrease the rank.

\begin{obs}\label{obs:other}
The following are graphs that are nucleation-free with implied non-edges, but are not encompassed by the schemes of this paper.
\begin{enumerate}
\item Ring of modified octahedral graphs with size 7, 8, 9, or 10. A modified octahedral graph is drawn in Fig. ~\ref{fig:octaForstarting}.
\begin{center}
\begin{figure}
\begin{center}
\includegraphics[width=.9\textwidth, height=.6\textwidth]{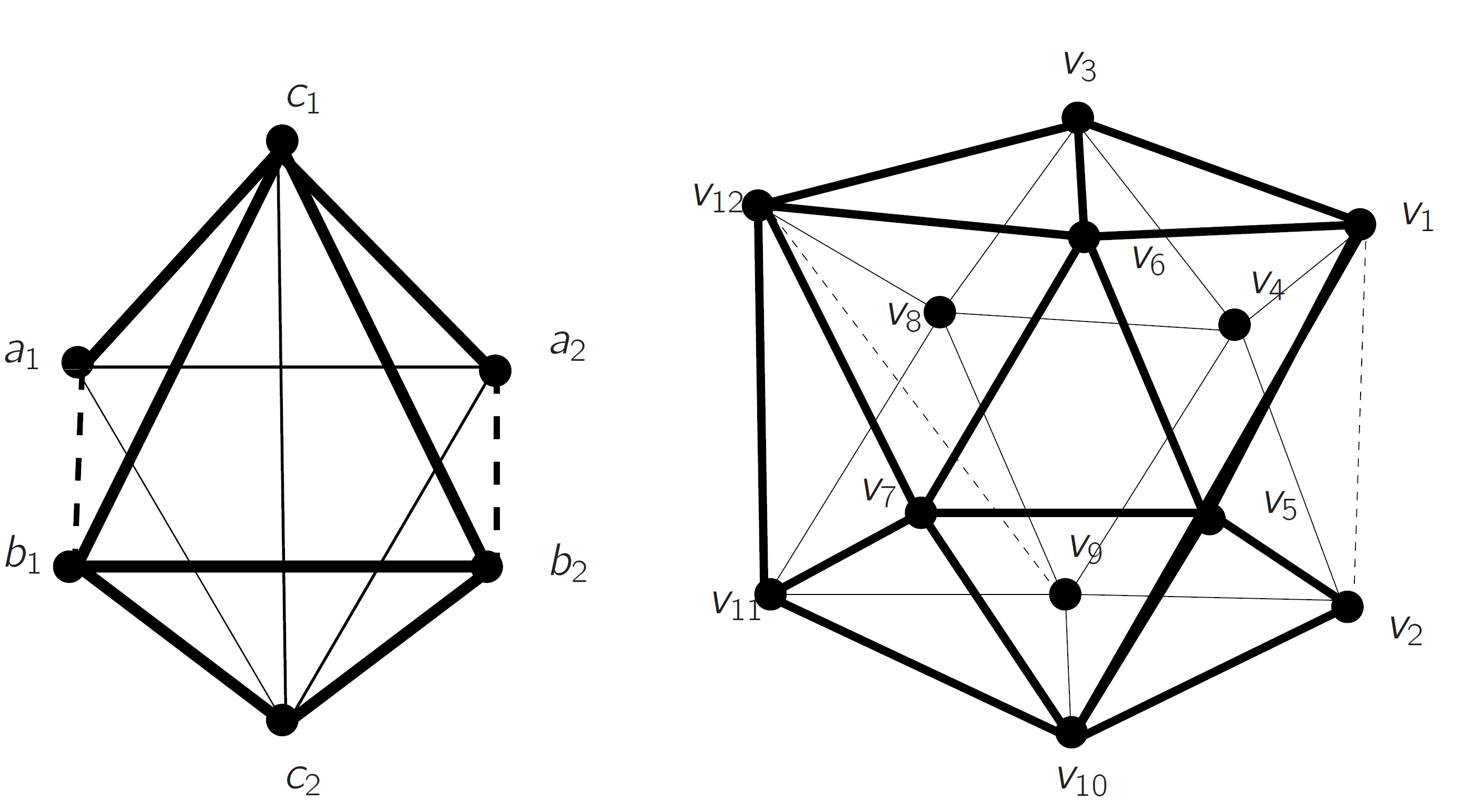}
\end{center}
\caption{On the left is another modified octahedral graph. This modified octahedral graph, when used as the building block of a ring graph, results in a ring that cannot be obtained using Henneberg constructions or all other known inductive methods from a ring of roofs. On the right is a modified icosahedral graph used to form a ring that is independent with implied non-edges. The dashed lines are the two hinge non-edges. }
\label{fig:octaForstarting}\label{fig:icosahedraon}
\end{figure}
\end{center}

\item Ring of modified icosahedral graphs with size 7, 8, 9, or 10. A modified icosahedral graph is drawn in Fig. ~\ref{fig:icosahedraon}.


\item Two body-sharing rings with icosahedra. Take two rings $R_1$ and $R_2$ and combine them in such a way that they share two or more covering subgraphs and one shared covering subgraph has two degrees-of-freedom, while the other shared covering subgraphs have 1 dof each. See Fig.~\ref{fig:ring_k12} for a simple example. We call this type of class a {\em body-sharing ring}.

\begin{center}
\begin{figure}
\begin{center}
\includegraphics[width=.5\textwidth]{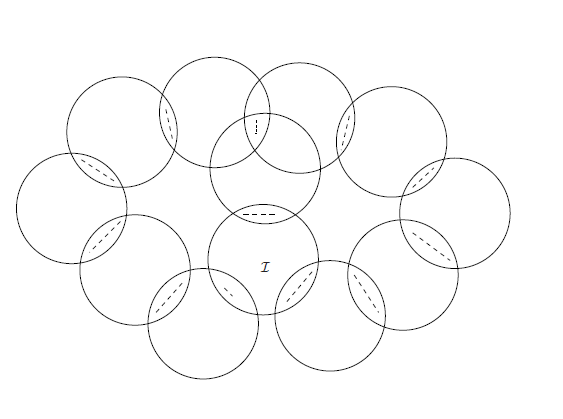}
\end{center}
\caption{Two rings of icosahedra share two icosahedra as in Fig. ~ \ref{fig:icosahedraon}. Each ring consists of $7$ icosahedra, represented by a circle in the figure. We drop two edges from $\I$ (namely $(v_1, v_2)$ and $(v_3, v_4)$ as in Fig. ~\ref{fig:icosahedraon}) and one edge ($(v_1, v_2)$) from all other icosahedra. We choose a non-edge $(v_9, v_{12})$ as another hinge non-edge for all icosahedra. I.e., $\I$ has 3 hinge non-edges, where one is a non-edge of icosahedra and the other two are dropped edges, while all other icosahedra have two hinge non-edges, where one is a non-edge of an icosahedron and the other is a dropped edge. The full rank of the graph is verified by symbolically computing the generic rank of a rigidity matrix of indeterminates for frameworks corresponding to the graph (for example, using Maple).}
\label{fig:ring_k12}
\end{figure}
\end{center}
\item Four body-sharing rings. See example in Fig. ~\ref{fig:otherScheme}.

\begin{figure}
\begin{center}
\includegraphics[width=.6\textwidth]{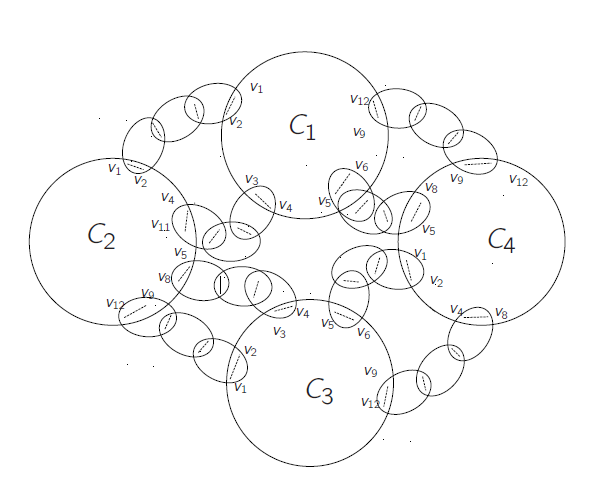}
\end{center}
\caption{Four body-sharing rings. The graph has four larger icosahedral graphs ($C_1$, $C_2$, $C_3$ and $C_4$ as in Fig.\ref{fig:icosahedraon}) represented with larger circles and $24$ smaller roofs represented with smaller ellipses. The dashed lines in the graph represent hinge non-edges. Note that $C_1$ and $C_3$ are modified icosahedral graphs whose edges $(v_3, v_4)$ and $(v_5, v_6)$ are dropped and used as hinge non-edges. They each have $27$ edges and thus have $3$ degrees-of-freedom. $C_2$ and $C_4$ are modified icosahedral graphs whose non-edges $(v_4, v_{11})$ and $(v_5, v_8)$ are used as hinge non-edges. They each have $29$ edges and thus have $1$ degree-of-freedom. The hinge non-edges of roofs are mentioned as earlier. The graph is clearly nucleation-free and it has $104$ vertices and $304$ edges. The full rank of the graph is verified by symbolically computing the generic rank of a rigidity matrix of indeterminates for frameworks corresponding to the graph (for example, using Maple). A 2-thin cover argument gives us that the dashed lines are all implied non-edges.  }\label{fig:otherScheme}
\end{figure}

\end{enumerate}
\end{obs}
\begin{proof}
We use the rank-sandwich proof technique. The rank upper bound argument for all examples is clear. The independence of those graphs is verified by symbolically computing the
generic rank of a  rigidity matrix of indeterminates for frameworks
corresponding to the example graphs (for example, using Maple). As mentioned earlier, we do not have to worry about numerical rounding issues distorting the rank
computation since such distortions can only decrease the rank.
\end{proof} 

\medskip\noindent
To extend Observation \ref{obs:other} to arbitrary sized rings, again, the rank upper bound is straightforward. To show independence, we would require proving a generalized version of roof-addition that permits adding general polyhedra.

\medskip\noindent
Next, we show how we can inductively apply Theorem \ref{thm:sums} and Theorem \ref{thm:secondInductiveApproach} to generate arbitrarily large nucleation-free, independent graphs with implied non-edges.

\begin{theorem}\label{thm:starting}
If a graph $H$ satisfies the starting graph condition for Theorem \ref{thm:secondInductiveApproach}, then applying $k$-sum or Henneberg-I constructions as in Theorem \ref{thm:sums} on $H$ gives a starting graph for Theorem \ref{thm:secondInductiveApproach}. If we apply roof-addition on $H$ according to Theorem \ref{thm:secondInductiveApproach}, we will get an output graph $G$ that again satisfies the starting graph condition for Theorem \ref{thm:secondInductiveApproach}.
\end{theorem}
\begin{proof} First, for $K$-sums on a starting graph $H$ and a nucleation-free, independent graph $G_1$, then resulting graph $G$ is independent and nucleation-free from Theorem \ref{thm:sums}. We can take the cover $\X$ of $H$ and extend it to a cover of $G$ by adding all pairs of vertices of $G_1$ to $\X$. It is easy to check that if $\X$ satisfies the rank bound condition on $H$, the new cover satisfies the rank bound condition on $G$.

Second, for Henneberg-I constructions on starting graph $H$, there are two cases. (1) The newly added vertex $u$ has all three neighbors inside one covering subgraph $C_i$. In this case, we can extend $C_i$ by adding $u$ to it and maintain all other covering subgraphs in the cover. The rank bound condition certainly holds in this case. (2) The three neighbors of $u$ are in different covering subgraphs. In this case, we can again take the cover $\X$ of $H$ and extend it to a cover of $G$ by adding three covering subgraphs that are three edges incident to $u$. It is easy to check that if $\X$ satisfies the rank bound condition on $H$, the new cover satisfies the rank bound condition on $G$.

Third, from Theorem \ref{thm:inductive}, we know apply roof-addition on starting graph $H$ maintains independence. It is obvious that nucleation-freeness is also maintained. Moreover, it is easy to check that applying roof-addition according to Theorem \ref{thm:secondInductiveApproach} maintains all other properties required for starting graphs.

Last, for all above schemes, it is clear that the cutting and distribution specified in Theorem \ref{thm:secondInductiveApproach} can be applied on any implied non-edge.
\end{proof}

While Theorem \ref{thm:starting} gives us a sufficient condition for starting graphs for Theorem \ref{thm:secondInductiveApproach}, it is not necessary. 

In this section, we showed an inductive method to construct independent graphs in Theorem \ref{thm:inductive} and illustrated how we can inductively apply Theorem \ref{thm:sums} and Theorem \ref{thm:secondInductiveApproach} to generate arbitrarily large nucleation-free, independent graphs with implied non-edges. Next we will turn to a construction scheme for nucleation-free dependent graphs, which uses the implied non-edges in the above constructed nucleation-free, independent graphs.

\section{Dependent Graphs with No Nucleus and Further Consequences}\label{sec:dependent}

In this section, we introduce a construction scheme (Scheme \ref{scheme:depenedent}) for nucleation-free dependent graphs. An illustration using rings of roofs is given in Fig.~\ref{fig:doubleRing}.

\begin{figure}[!htbp]
\centering
\scalebox{0.4}[0.35]{
\begin{tabular}{c}
\includegraphics[width=1.2\textwidth]{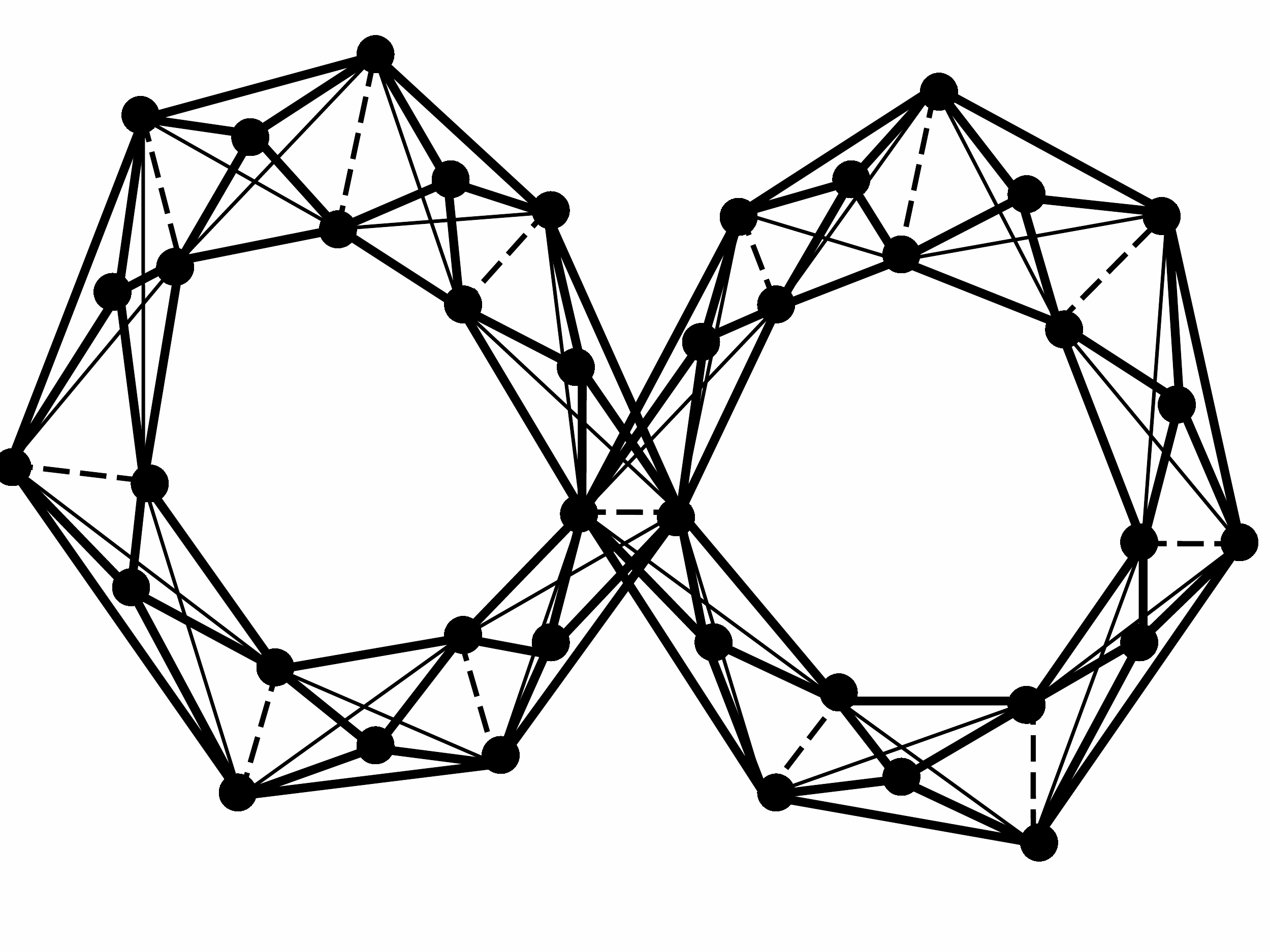}
\includegraphics[width=1.2\textwidth]{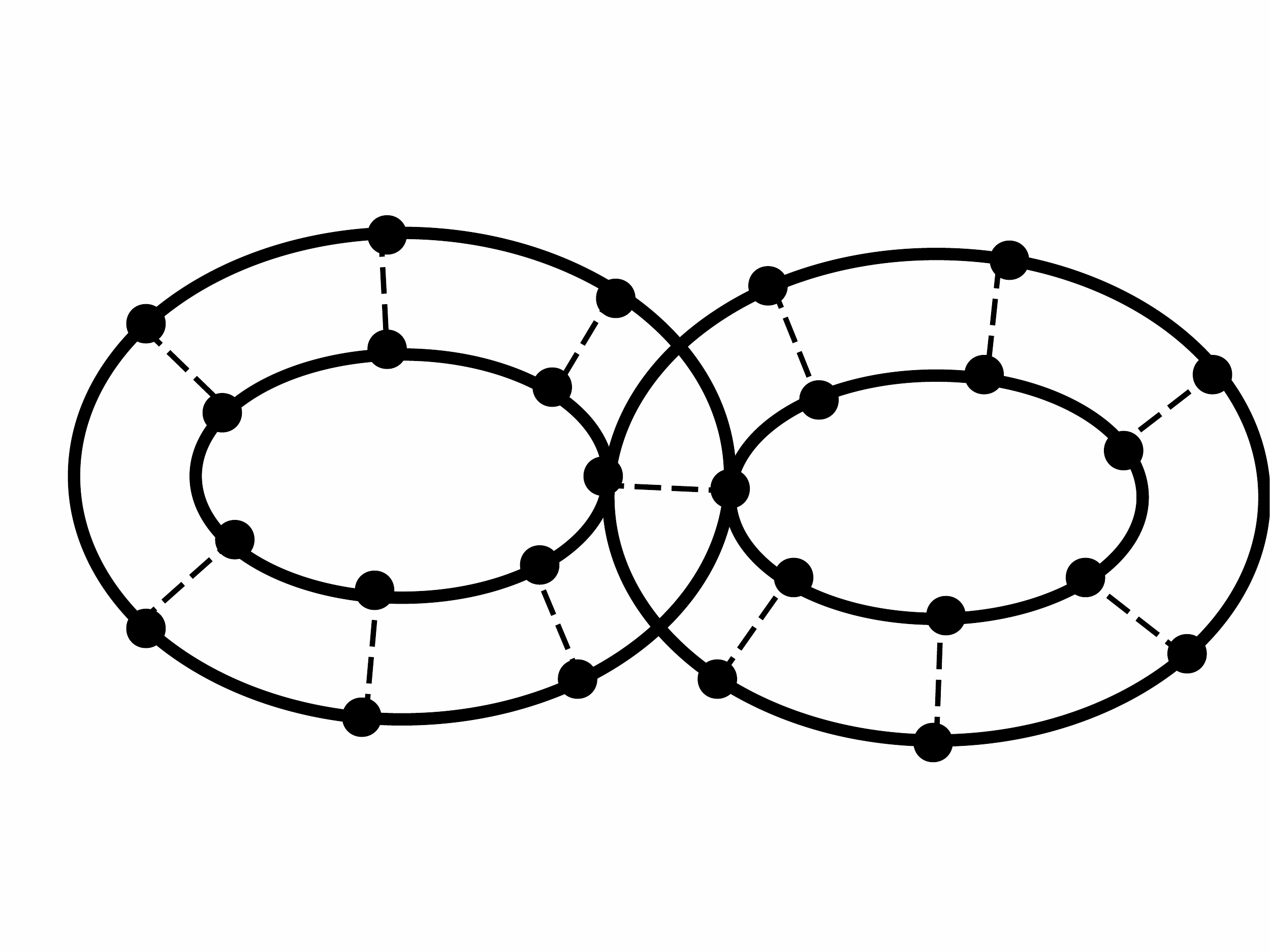}
\end{tabular}}
\caption{A double-ring of $14$ roofs consisting of two rings of $7$ roofs. The two rings share a common non-edge. The figure on the left shows the double-ring of roofs and the figure on the right shows the schematic of any pair of nucleation-free, independent graph with a shared implied non-edge. Since each of the two parts implies the shared hinge non-edge, as in the double-banana example in Fig.~\ref{fig:doubleBanana}, the shared hinge non-edge is double-implied and hence the composite of the two parts is dependent.}\label{fig:doubleRing}
\end{figure}

\begin{scheme}[Graph-combination]\label{scheme:depenedent}\hfill

{\bf Input graphs:} $G_1$ with implied non-edge $\{a_1, b_1\}$ and $G_2$ with implied non-edge $\{a_2, b_2\}$

\medskip
{\bf Output graph:} Graph $G$ after joining $G_1$ and $G_2$ by identifying $a_1$ with $a_2$ and $b_1$ with $b_2$.
\end{scheme}

\begin{theorem}\label{thm:joinIndep}
For graph-combination scheme (Scheme \ref{scheme:depenedent}), if $G_1$ and $G_2$ are both nucleation-free, then $G$ is nucleation-free dependent. If, in addition, $G_1\cup$ $\{a_1, b_1\}$ and $G_2$ $\cup$ $\{a_2, b_2\}$ are both rigidity circuits, then the output graph $G$ is a rigidity circuit.
\end{theorem}
\begin{proof}(of Theorem \ref{thm:joinIndep}). We know both $G_1$ and $G_2$ imply the shared non-edge and they each have $1$ dof. After glueing $G_1$ and $G_2$, there must be a dependence in the output graph $G$. 

If in addition $G_1\cup$ $\{a_1, b_1\}$ and $G_2$ $\cup$ $\{a_2, b_2\}$ are both rigidity circuits, to show the output graph $G$ is a rigidity circuit, we can drop any edge from $G$ to obtain $G^\prime$ and show that $G^\prime$ is independent. Without loss of generality, we drop an edge $e_1$ from $G_1$ to obtain $G_1^\prime$. Since $G_1\cup$ $\{a_1, b_1\}$ is a rigidity circuit, we know $G_1^\prime$ $\cup$ $\{a_1, b_1\}$ is independent.
Since $G_2$ $\cup$ $\{a_2, b_2\}$ is a rigidity circuit, if we choose an arbitrary edge $e_2$ on $G_2$, then (1) $G_2$ $\setminus$ $\{e_2\}$ $\cup$ $\{a_2, b_2\}$ is independent; and (2) the linear span of the edges in $G_2$ is equal to the linear span of the edges in $G_2$ $\setminus$ $\{e_2\}$ $\cup$ $\{a_2, b_2\}$. It follows that the linear span of $G^\prime$ is equal to the linear span of $G^{\prime\prime}$ $:=G^\prime$ $\setminus$ $\{e_2\}$ $\cup$ $\{a_2, b_2\}$. Moreover, we can easily see that $G^{\prime\prime}$ is a $2$-sum of $G_1^\prime$ $\cup$ $\{a_1, b_1\}$ and $G_2$ $\setminus$ $\{e_2\}$ $\cup$ $\{a_2, b_2\}$, both of which are independent. So we can use a similar argument as in the proof of Theorem \ref{thm:sums} to show that $G^{\prime\prime}$ is independent, which in turn means $G^\prime$ is independent.
Thus it follows that the output graph $G$ for Scheme \ref{scheme:depenedent} is minimally dependent.
\end{proof} 



\begin{figure}[!htbp]
\centering
\scalebox{0.4}[0.35]{
\begin{tabular}{c}
\includegraphics[width=1.2\textwidth]{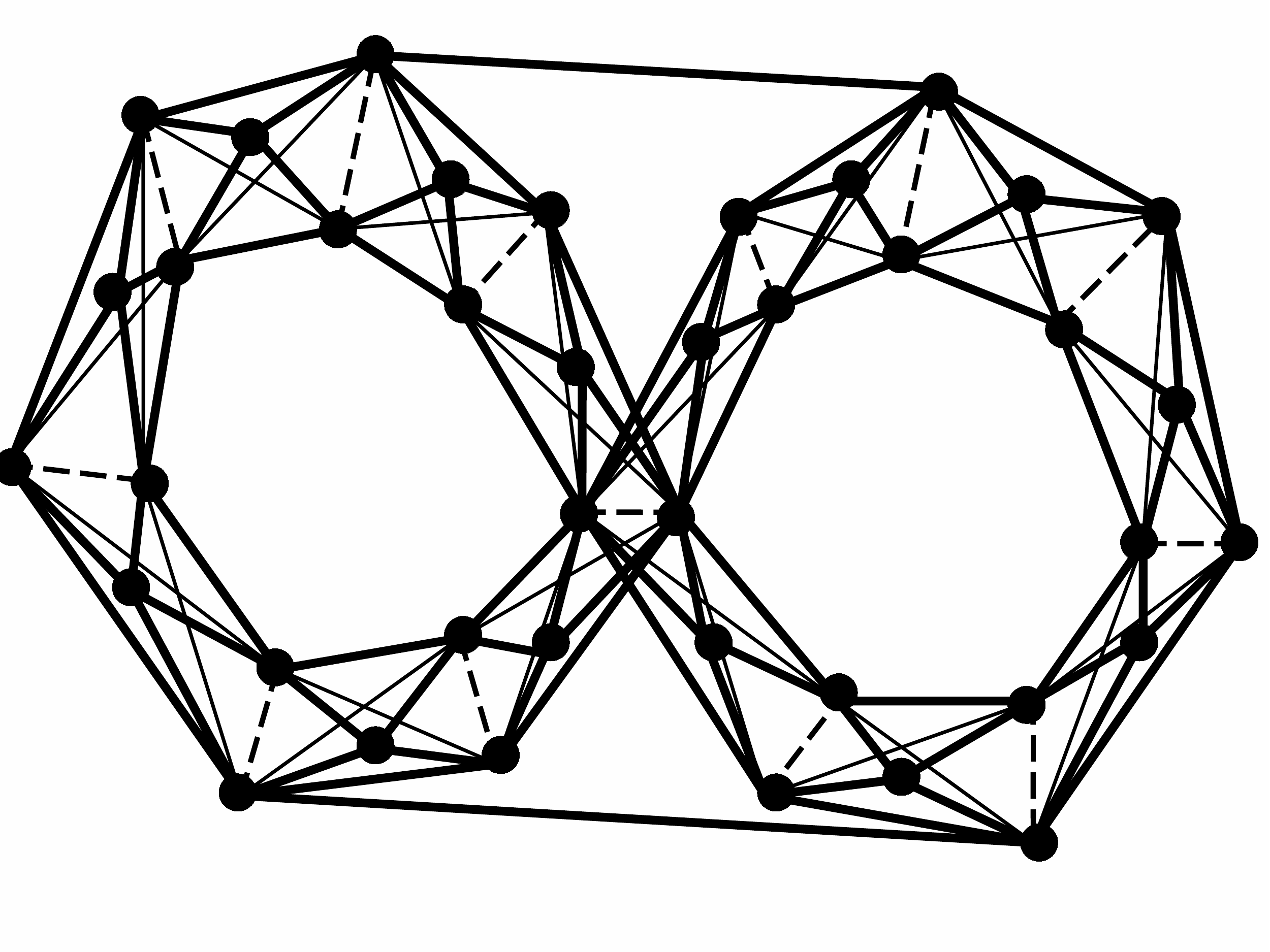}
\includegraphics[width=1.2\textwidth]{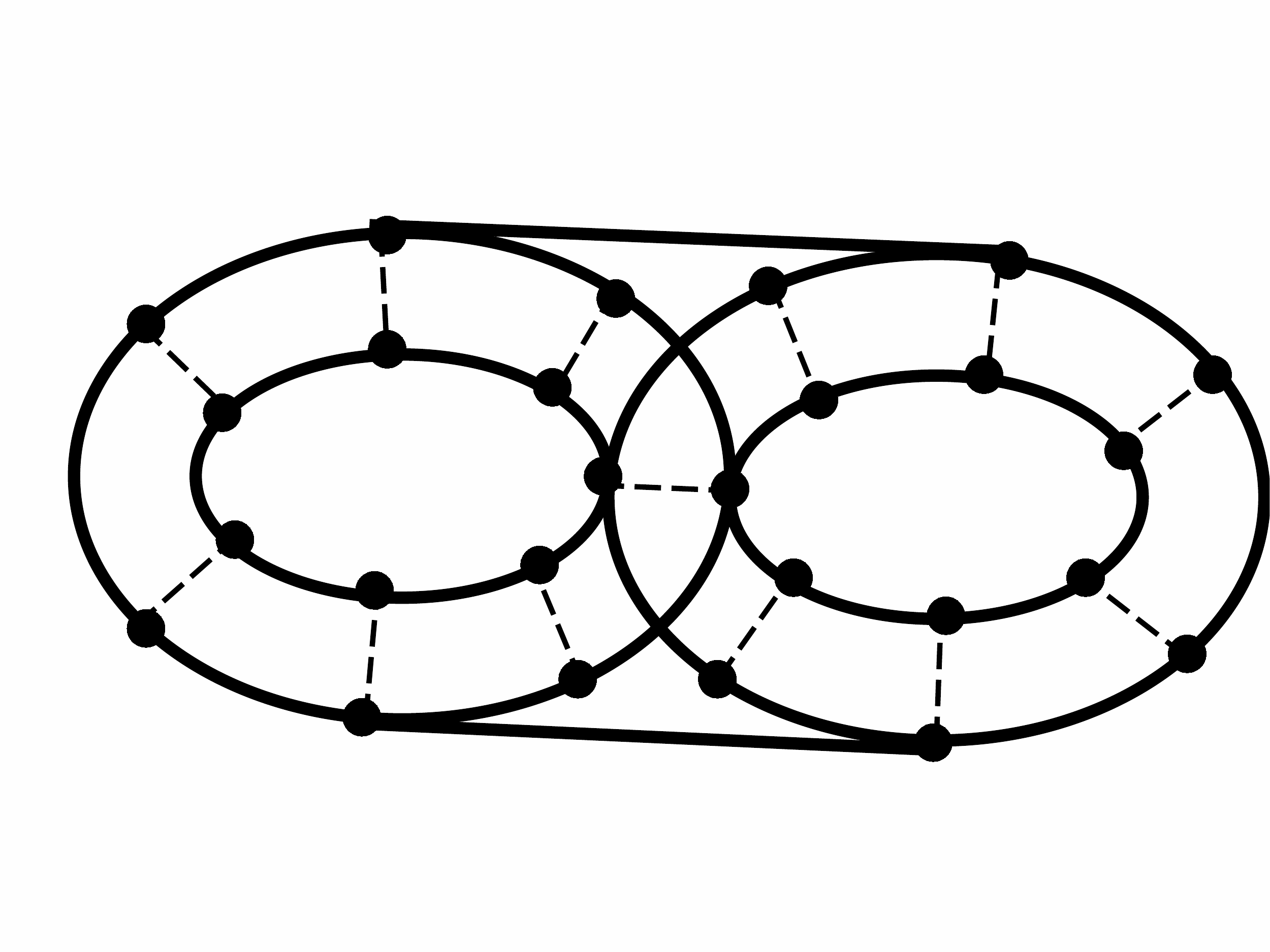}
\end{tabular}}
\caption{A braced double-ring of $14$ roofs: it consists of two rings of $7$ roofs and two extra bars. The two rings in the graph are connected via a common non-edge. The figure on the left shows the geometric structure of the braced double-ring and the figure on the right shows the schematic. The braced double ring has enough edges to be minimally rigid but is in fact flexible from Theorem \ref{thm:joinIndep}.}\label{fig:doubleRingBar}
\end{figure}

Sitharam and Zhou \cite{sitharam:zhou:tractableADG:2004} gave several examples where Maxwell's counting condition was insufficient for rigidity, all of which satisfy the nucleation property. Using a combinatorial notion capturing the recursive nature of nucleation (called {\it module-rigidity}), they propose an algorithm which is a tractable characterization of generic independence in a large class of such graphs by using the presence of rigid nuclei. 
It has been an open problem whether this algorithm can fail to detect 3D independence and rigidity, i.e. whether module-rigidity coincides or not with 3D rigidity. The following corollary settles the question (in the negative), i.e., there are graphs that are not rigid but are module-rigid.
\begin{corollary}\label{cor2}
The flexible braced double-ring in Fig.~\ref{fig:doubleRingBar} is {\em module-rigid} according to the definition in \cite{sitharam:zhou:tractableADG:2004}. Therefore, module-rigidity does not coincide with 3D rigidity.
\end{corollary}

\begin{proof}(of Corollary \ref{cor2}). When a graph has no nucleation, Sitharam and Zhou's \cite{sitharam:zhou:tractableADG:2004}
algorithm reduces to $(3, 6)$-sparsity check. For a braced double-ring, the graph will be declared module-rigid. On the other hand, from Theorem \ref{thm:joinIndep}, we know a double-ring is dependent, thus the braced double-ring is also dependent. Together with the fact that the braced double-ring has minimum number of edges to be rigid, we know its rank (which is smaller than its size due to dependence) cannot be enough to be rigid.
%
\end{proof} 



\section{Conclusions and Open Problems}
\label{sec:conclusions}
In this paper we provided general inductive construction schemes for nucleation-free (independent) graphs with implied non-edges, nucleation-free dependent graphs and nucleation-free circuits.
Besides settling problems posed previously in the literature, this work extends the repertoire of useful examples that elucidate
the obstructions to obtaining combinatorial characterizations of 3D rigidity. We have provided two proof techniques for showing implied non-edges in nucleation-free graphs. 

\medskip\noindent
Table \ref{tab:warmup} sums up the constructions we have presented in this paper, and their associated proof techniques. 

\begin{table}\large

\resizebox{\columnwidth}{!}{
\begin{tabular}{|l|l|l|l|l|l|}
\hline
\multicolumn{2}{|c|}{\backslashbox[3cm]{B}{A} } & \begin{tabular}[c]{@{}c@{}}Flex-sign \\ Ring \end{tabular}   & \begin{tabular}[c]{@{}c@{}}Henneberg extender \\ Ring\end{tabular} & \begin{tabular}[c]{@{}c@{}}Standard- \\ scheme\end{tabular} & Roof-addition\\ \hline
\multicolumn{2}{|c|}{Flex-sign} & Theorem \ref{thm:specialmotion} & $-$ & $-$ &$-$ \\ \hline
\multirow{2}{*}{\begin{tabular}[c]{@{}c@{}}Rank-\\sandwich\end{tabular} } & (i) Independence & $-$ & (i) Theorem \ref{thm:Henneberg} & (i) Theorem \ref{thm:sums} & (i) Theorem \ref{thm:inductive}\\ 
& (ii) Rank upper-bound & $-$& (ii) C & (ii) Not necessary  & (ii) C: Theorem \ref{thm:secondInductiveApproach} \\ \hline
  \end{tabular}
}
\caption{Construction schemes for nucleation-free, independent graphs with implied non-edges. Here, A denotes construction schemes,  B denotes the proof techniques and C denotes 2-thin cover/body-hinge argument.}\label{tab:warmup}
\end{table}

One open problem to extend the application of the first proof technique is to find other graphs that satisfy the expansion/contraction property. Another interesting open problem  to extend the application of our rank-sandwich proof technique is to find other construction schemes for independent graphs and other techniques besides 2-thin cover argument/body-hinge argument to prove rank upper bounds.

To complete our understanding of nucleation-free graphs with implied non-edges, the next step is to study examples extending those in Observation \ref{obs:different} that cannot be obtained by any of our construction schemes. 

Another interesting open problem is to extend our inductive construction for independent graphs to an inductive construction for {\em isostatic} graphs (independent and minimally rigid). In order to do that, we need to add two more edges in the roof pasting step. One possible way is to add two more edges $(c^\prime, a_1)$, $(c, b_2)$ (or $(c^\prime, b_1)$, $(c, a_2))$. We note that our current method to show independence in Theorem \ref{thm:inductive} would fail, since Claim \ref{clm:nonzero} fails in that there is a non-zero stress on the added part. However, if we can show that if there is a generic circuit in the new graph $G$, then there is one $\{(w_1, t_1), \ldots (w_n, t_n)\}$ that remains a dependence for the non-generic position $\p$ used in Theorem \ref{thm:inductive}, i.e., there exists non-zero stresses $\{s_1, \ldots, s_n\}$ s.t. $\sum_i s_i (\p(w_i)-\p(t_i))$ $=$ $0$, then we can simply use Claim \ref{clm:stressA} to complete our proof.

\begin{appendix}
\section{Issues in Tay's paper \cite{taybar:1993}}\label{sec:Taycounter}
As mentioned previously in Section \ref{sec:introduction}, nucleation-free rigidity circuits with implied non-edges have been conjectured and written down by many (\cite{taybar:1993}, \cite{JacksonJordanrank:2006}). However, to the best of our knowledge, we are the first to give proofs. In particular, in \cite{taybar:1993}, Tay claimed a class of flexible rigidity circuits without any nuclei. One of his examples, {\em $n$-butterflies}, in which he claimed existence of implied non-edges, is the same as our warm-up example graphs, {\em ring of roofs}. In \cite{taybar:1993}, Tay presented a proof of the independence of ring of roofs. His proof is based on the Proposition 4.6 in the paper, which shows the existence of implied non-edges in the ring. However, the argument he made in the proof is imprecise at best. In Proposition 4.6, he first took a chain of graphs $G_1, G_2, \ldots G_n$ which is known to be a 3D rigidity circuit as a whole, and then he closed the chain and subsequently removed the two joining edges $(p, q)$ and $(r, s)$. He stated, if the stresses ($\lambda_{pq}$ and $\lambda_{rs}$) on those two edges cancel out, ``one can keep $\lambda_{pq}$ fixed and change the value of $\lambda_{rs}$ by changing the position of either $a_1$ or $b_1$'' ($a_1$, $b_1$ are the two vertices shared by the first two subgraphs $G_1$ and $G_2$ of the chain).

However, he did not mention how to change the position of $a_1$ or $b_1$. Thus one can find some example that a small change of the position would not affect the stresses on $pq$ or $rs$. For example, in Fig.~\ref{fig:counterTay1}, we have a chain of graphs where $G_1$ happens to be the union of the rest of the subgraphs in the chain. The whole graph is realized in a position where it is symmetric along the line $a_1b_1$. If we move $a_1$ or $b_1$ along the current line $a_1b_1$, the stresses on $pq$ and $rs$ should always cancel out.
\begin{figure}[!htbp]
\centering
\scalebox{0.4}[0.35]{
\includegraphics{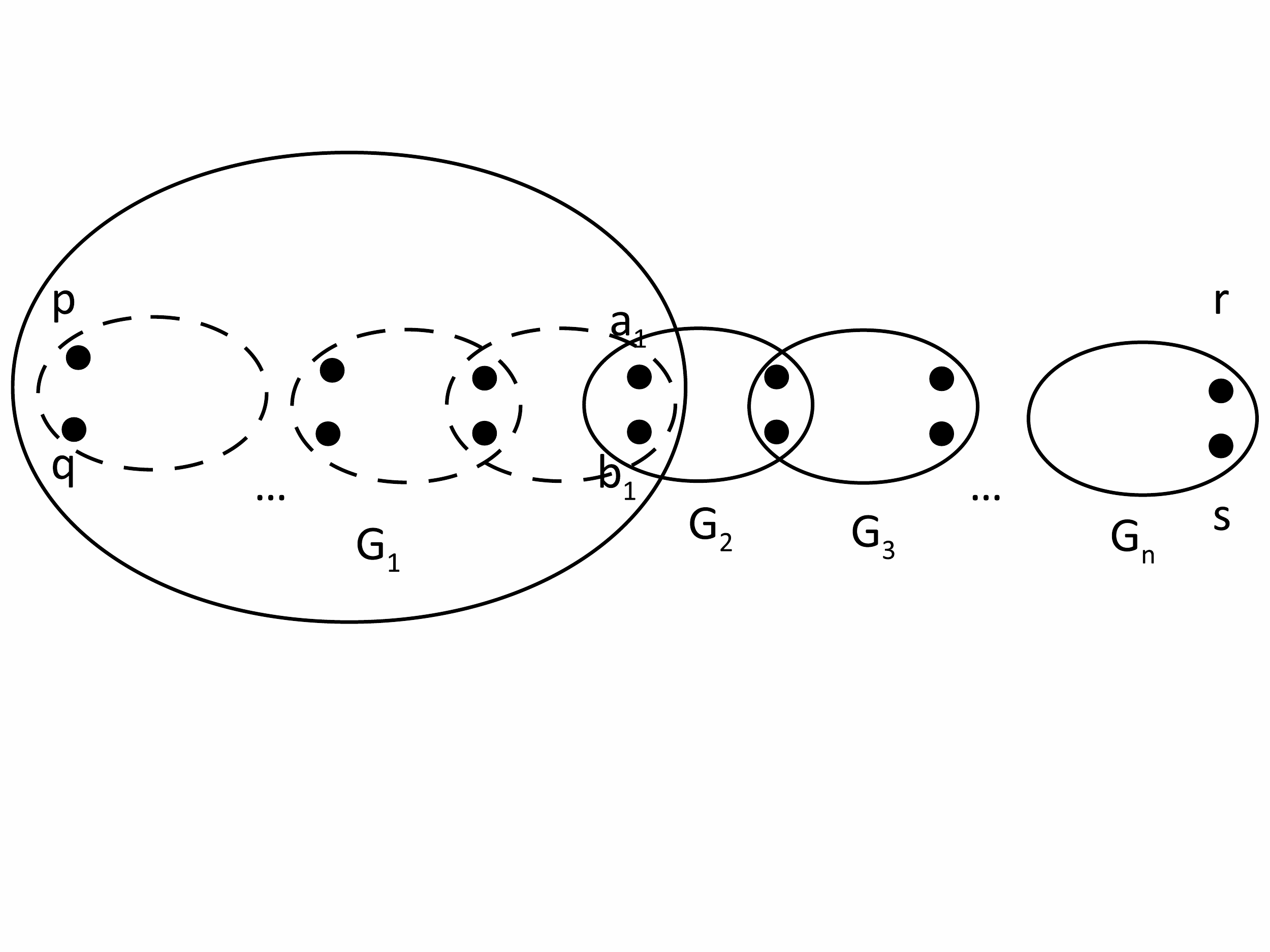}}
\caption{A counterexample to Tay's proof: the first subgraph is the union of the rest and we put them in such a position that the chain is symmetric along $a_1b_1$.}\label{fig:counterTay1}
\end{figure}

Another potential hole is that he did not specify the starting realization of the graph. So it is possible that in some realization, not necessarily generic, of such a chain graph, no matter how one alters the position of $a_1$ or $b_1$, the stresses of $pq$ and $rs$ remains opposite of each one in magnitude. For example, in Fig.~\ref{fig:counterTay2}, $a_1$ and $b_1$ are not adjacent to $p$ or $q$. Then we put vertices that are adjacent to $p$ or $q$ at the same position (indicated as $v$ in Fig.~\ref{fig:counterTay2}). Then we put $(p,v)$ in such a position that the angle between $pv$ and $pq$ is $45$ degrees. Similarly make the angle between $qv$ and $pq$ to be $45$ degrees. Then it is easy to obtain that the stresses on $(p, q)$ is zero. Do the same for $(r, s)$. No matter how we move $a_1$ or $b_1$, the stress, $\lambda_{pq}$, is always zero, and so is $\lambda_{rs}$.

\begin{figure}[!htbp]
\centering
\scalebox{0.4}[0.35]{
\includegraphics{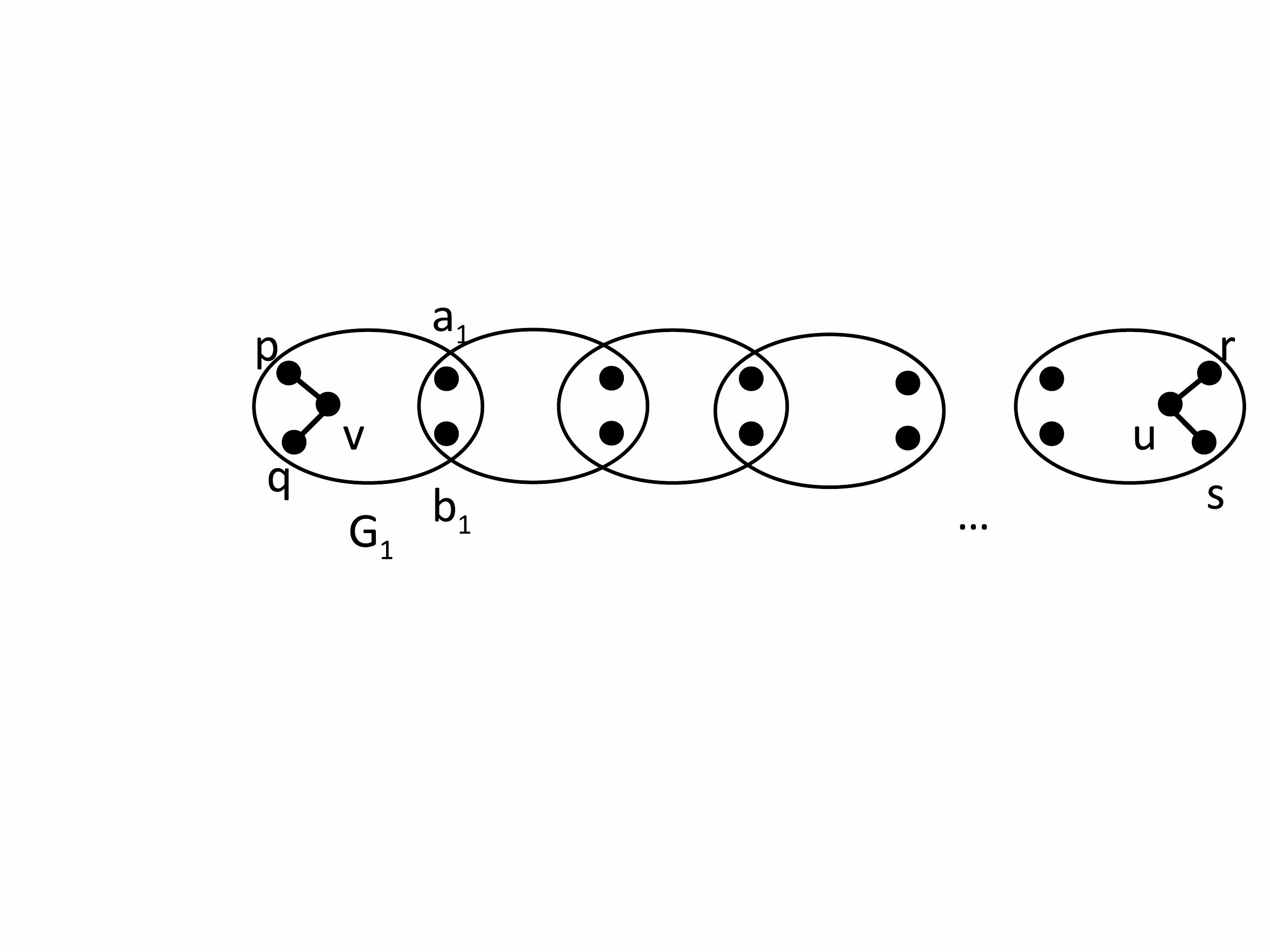}}
\caption{Another counterexample to Tay's proof: we put all vertices adjacent to $p$ or $q$ at the same position as $v$ and make sure $pq, pv,qv$ forms a isosceles right triangle. It is not hard to show that the stress $\lambda_{pq}$ is always zero. We do the same for the neighbors of $r$ and $s$ thus the stresses $\lambda{pq}$ and $\lambda{rs}$ always cancel out.}\label{fig:counterTay2}
\end{figure}

\end{appendix}

\bibliographystyle{plain}
\bibliography{biblio}

\begin{thebibliography}{10}

\bibitem{bib:counter2}
A.~Drieding A.~Dress and H.~Haegi.
\newblock Classification of mobile molecules by category theory.
\newblock {\em Symmetries and properties of non-rigid molecules: A
  comprehensive studys}, pages 39--58, 1983.

\bibitem{BelkConnelly07}
Maria Belk and Robert Connelly.
\newblock Realizability of graphs.
\newblock {\em Discrete {\&} Computational Geometry}, 37(2):125--137, 2007.

\bibitem{connelly:rigidityAndEnergy:1982}
Robert Connelly.
\newblock Rigidity and energy.
\newblock {\em Inventiones Mathematicae}, 66:11--33, 1982.

\bibitem{graver:servatius:rigidityBook:1993}
Jack Graver, Brigitte Servatius, and Herman Servatius.
\newblock {\em Combinatorial rigidity}.
\newblock Graduate Studies in Mathematics. American Mathematical Society, 1993.

\bibitem{bib:counter1}
Bill Jackson and Tibor Jordan.
\newblock The dress conjectures on rank in the 3-dimensional rigidity matroid.
\newblock {\em Advances in Applied Mathematics}, 35:355--367, oct 2005.

\bibitem{JacksonJordanrank:2006}
Bill Jackson and Tibor Jord\'{a}n.
\newblock On the rank function of the 3-dimensional rigidity matroid.
\newblock {\em Int. J. Comput. Geom. Appl}, 16(5/6):415--429, December 2006.

\bibitem{jackson:jordan:szabadka:GloballyLinked:2006}
Bill Jackson, Tibor Jord\'{a}n, and Zolt{\'a}n Szabadka.
\newblock Globally linked pairs of vertices in equivalent realizations of
  graphs.
\newblock {\em Discrete Comput. Geom.}, 35(3):493--512, 2006.

\bibitem{laman:rigidity:1970}
Gerard Laman.
\newblock On graphs and rigidity of plane skeletal structures.
\newblock {\em Journal of Engineering Mathematics}, 4:331--340, 1970.

\bibitem{lovasz:yemini:rigidity:1982}
L{\'a}szl{\'o} Lov{\'{a}}sz and Y.~Yemini.
\newblock On generic rigidity in the plane.
\newblock {\em SIAM J. Algebraic and Discrete Methods}, 3(1):91--98, 1982.

\bibitem{maxwell:equilibrium:1864}
James~Clerk Maxwell.
\newblock On the calculation of the equilibrium and stiffness of frames.
\newblock {\em Philosophical Magazine}, 27:294--299, 1864.

\bibitem{recski:networkRigidity:1984}
Andr{\'{a}}s Recski.
\newblock A network theory approach to the rigidity of skeletal structures
  {II}. {L}aman's theorem and topological formulae.
\newblock {\em Discrete Applied Math}, 8:63--68, 1984.

\bibitem{bib:wellformed}
M.~Sitharam.
\newblock Characterizing well-formed systems of incidences for resolving
  collections of rigid bodies.
\newblock {\em IJCGA}, 16(5-6):591--615, 2006.

\bibitem{bib:opt}
M.~Sitharam, J.~Peters, and Yong Zhou.
\newblock Optimized parametrization of systems of incidences between rigid
  bodies.
\newblock {\em Journal of Symbolic Computation}, 45:481--498, feb 2010.
\newblock YJSCO1143.

\bibitem{bib:SitharamFrontier}
Meera Sitharam.
\newblock Frontier, opensource gnu geometric constraint solver: Version 1
  (2001) for general 2d systems; version 2 (2002) for 2d and some 3d systems;
  version 3 (2003) for general 2d and 3d systems.
\newblock {\url {http://www.cise.ufl.edu/~sitharam}}.

\bibitem{Sitharam:2010:CGC}
Meera Sitharam and Heping Gao.
\newblock Characterizing graphs with convex and connected cayley configuration
  spaces.
\newblock {\em Discrete Comput. Geom.}, 43(3):594--625, April 2010.

\bibitem{sitharam:zhou:tractableADG:2004}
Meera Sitharam and Yong Zhou.
\newblock A tractable, approximate characterization of combinatorial rigidity
  in 3d.
\newblock In {\em Abstract 5th Automated Deduction in Geometry (ADG)}, 2004.
\newblock http://www.cise.ufl.edu/~sitharam/module.pdf.

\bibitem{streinu:pseudoTriang:dcg:2005}
Ileana Streinu.
\newblock Pseudo-triangulations, rigidity and motion planning.
\newblock {\em Discrete and Computational Geometry}, 34(4):587--635, November
  2005.

\bibitem{streinu:whiteley:origami:2005}
Ileana Streinu and Walter Whiteley.
\newblock Single-vertex origami and spherical expansive motions.
\newblock In Jin Akiyama and M.~Kano, editors, {\em Proc. Japan Conf. Discrete
  and Computational Geometry (JCDCG 2004)}, volume 3742 of {\em Lecture Notes
  in Computer Science}, pages 161--173, Tokai University, Tokyo, 2005. Springer
  Verlag.

\bibitem{bib:TayWhiteley85}
T.~Tay and W.~Whiteley.
\newblock Generating isostatic graphs,.
\newblock {\em Structural Topology}, 11:21--69, 1985.

\bibitem{tay:rigidityMultigraphs-I:1984}
Tiong-Seng Tay.
\newblock Rigidity of multigraphs {I}: linking rigid bodies in n-space.
\newblock {\em Journal of Combinatorial Theory, Series B}, 26:95--112, 1984.

\bibitem{tay:proofLaman:1993}
Tiong-Seng Tay.
\newblock A new proof of {L}aman's theorem.
\newblock {\em Graphs and Combinatorics}, 9:365--370, 1993.

\bibitem{taybar:1993}
Tiong-Seng Tay.
\newblock On generically dependent bar frameworks in space.
\newblock {\em Struct. Topol}, 1993.

\bibitem{bib:counter3}
T.S. Tay and W.~Whiteley.
\newblock Recent advances in the generic rigidity of structures.
\newblock {\em Structural Topology}, 9:31--38, 1984.

\bibitem{white:whiteley:algebraicGeometryFrameworks:1987}
Neil White and Walter Whiteley.
\newblock The algebraic geometry of motions of bar-and-body frameworks.
\newblock {\em SIAM Journal of Algebraic Discrete Methods}, 8:1--32, 1987.

\bibitem{WhiteleyVertexSplitting1990}
W.~Whiteley.
\newblock Vertex splitting in isostatic frameworks.
\newblock {\em Structural Topology}, 16:23--30, 1991.

\end{thebibliography}

\end{document}